\documentclass[10pt]{article}
\usepackage{amsmath}
\usepackage[utf8]{inputenc}
\usepackage[T1]{fontenc}
\usepackage{accents}

\usepackage[margin=1in]{geometry}
\usepackage{bm}
\usepackage{amssymb}
\usepackage[colorlinks]{hyperref}
\usepackage{subcaption}
\usepackage{tikz}
\usepackage{float}
\usepackage[numbers,sort&compress]{natbib}
\usepackage{booktabs}
\usepackage{threeparttable}
\usepackage{caption}
\usepackage[ruled]{algorithm2e}
\usepackage{algorithmic}
\usepackage{enumitem}
\usepackage{authblk}
\usepackage[multiple]{footmisc}

\usepackage{titlesec}
\newcommand{\ubar}[1]{\underaccent{\bar}{#1}}

\DeclareMathOperator*{\argmax}{arg\,max}
\usepackage{amsthm}

\newtheorem{theorem}{Theorem}

\newtheorem{proposition}[theorem]{Proposition}

\newtheorem{corollary}{Corollary}[theorem]

\pdfoutput=1

\title{Ensemble approximate control variate estimators: Applications to multi-fidelity importance sampling}
\author[a]{Trung Pham\footnote{Corresponding authors.}\footnote{Email addresses: \href{mailto:trungp@umich.edu}{trungp@umich.edu} (T. Pham), \href{mailto:goroda@umich.edu}{goroda@umich.edu} (A.A. Gorodetsky)}}
\author[a]{Alex A. Gorodetsky$^{*\dagger}$}
\affil[a]{\small Department of Aerospace Engineering, University of Michigan, Ann Arbor, MI, 48109, USA}
\date{\vspace{-7ex}}
\newcommand{\reals}{\mathbb{R}}

\newcommand{\prob}{P}
\newcommand{\ee}[2]{\mathbb{E}_{#1}\left[#2\right]}

\newcommand{\vvar}[2]{\mathbb{V}\text{ar}_{#1}\left[#2\right]}
\newcommand{\cov}[2]{\mathbb{C}\text{ov}_{#1}\left[#2\right]}

\providecommand{\keywords}[1] {
    \small	
    \textbf{Keywords.} #1
}

\begin{document}

\maketitle
\begin{abstract}
    The recent growth in multi-fidelity uncertainty quantification has given rise to a large set of variance reduction techniques that leverage information from model ensembles to provide variance reduction for estimates of the statistics of a high-fidelity model. In this paper we provide two contributions: (1) we utilize an ensemble estimator to account for uncertainties in the optimal weights of approximate control variate (ACV) approaches and derive lower bounds on the number of samples required to guarantee variance reduction; and (2) we extend an existing multi-fidelity importance sampling (MFIS) scheme to leverage control variates. As such we make significant progress towards both increasing the practicality of approximate control variates---for instance, by accounting for the effect of pilot samples---and using multi-fidelity approaches more effectively for estimating low-probability events. The numerical results indicate our hybrid MFIS-ACV estimator achieves up to 50\% improvement in variance reduction over the existing state-of-the-art MFIS estimator, which had already shown outstanding convergence rate compared to the Monte Carlo method, on several problems of computational mechanics.
\end{abstract}

\keywords{multi-fidelity, uncertainty quantification, approximate control variates, importance sampling, rare-event simulation}

\section{Introduction}

        This paper develops an advancement to the approximate control variate (ACV)~\cite{gorodetsky_generalized_2020} approach for variance reduction in uncertainty quantification applications where multiple models with varying qualities and computational costs are available. Specifically, we analyze the affect of using estimated control variate weights within the ACV on variance reduction, and we provide conditions under which variance reduction can still be guaranteed in these cases.  Multi-fidelity approaches for uncertainty quantification have recently seen significant adoption across wide varying domains where expensive simulations are required for accurate predictions. These domains include heat transfer problems \cite{doostan_bi-fidelity_2016}, aerospace design \cite{geraci_multifidelity_2017}, optimization under uncertainty \cite{chaudhuri_multifidelity_2018,ng_multifidelity_2014}, and ensemble of computer simulator outputs \cite{goh_prediction_2013}. A survey of multi-fidelity methods is presented in \cite{peherstorfer_survey_2018}.

    As with single fidelity UQ, multi-fidelity (MF) UQ techniques can leverage both surrogate and sampling-based algorithmic approaches. While surrogate-based techniques are plentiful \cite{peherstorfer_multifidelity_2016,peherstorfer_combining_2017,kramer_multifidelity_2019} we focus on sampling approaches that are often both more flexible to leverage and also provide the foundation for many surrogate approaches (e.g., estimating where to obtain more data). The basis of a majority of sampling approaches is the usage of Monte Carlo (MC) simulation~\cite{fishman2014monte} to estimate the output statistics. The primary advantage of MC over surrogate approaches is that it does not impose any requirement on the smoothness of the forward model, and its accuracy and convergence rate are independent of the model dimension. Nevertheless, its convergence rate is also slow, demanding a large number of model evaluations to reach the satisfactory accuracy. This computational cost can be prohibitive for many practical problems with expensive simulation models. A straightforward error analysis in \cite{glasserman2004monte} reveals that the efficiency of MC simulation can be greatly improved by variance reduction methods, which, as explained in \cite{rubinstein2016simulation}, ``can be viewed as a means of utilizing known information about the model in order to obtain more accurate estimators of its performance.'' Two such methods are heavily used for quantifying uncertainty: importance sampling (IS) and control variates (CV). We explore the adaptation and advancement of certain aspects of these two approaches to multi-fidelity uncertainty quantification problems in this paper.
    
    Adapting standard statistical approaches for variance reduction in the context of UQ problems must address special challenges. The primary challenge is that the relationships about and between low-fidelity and high-fidelity models are unknown or imprecisely encoded. For instance, CV techniques require the low-fidelity models to have known means and known covariance amongst models. However, the low-fidelity information sources in uncertainty quantification typically are in the form of simulation models---their statistics are not known but simply easier to compute than the high-fidelity model. Thus, algorithms that adapt IS and CV to these problems must determine and model the relationship between such information sources.
    
            A number of multi-fidelity techniques leveraging importance sampling have also been proposed. Importance sampling approaches generate weighted samples from a biasing distributions. A prudent choice of biasing distribution can lead to a drastic reduction of computation cost~ \cite{srinivasan2013importance,bucklew2004introduction}. MFUQ techniques that use IS are based on the idea that the distributions of low-fidelity quantity of interest (QoI) are closely related to the high-fidelity QoI, even if their pointwise evaluations have errors. In \cite{peherstorfer_multifidelity_2016} the IS density is constructed from a single surrogate model built on the high-fidelity model, and the multi-fidelity importance sampling (MFIS) estimator is basically an IS estimator using that surrogate-based density. In \cite{peherstorfer_combining_2017} multiple low-fidelity surrogate-based IS densities are aggregated to derive an estimator following the mixed IS approach~\cite{art_owen_2000}. In \cite{kramer_multifidelity_2019} multiple surrogate-based IS
            estimators are fused into a weighed ensemble estimator, and the weights are determined through minimizing the variance of the fused estimator.
            Though these estimators have achieved impressive speedups, further improvement can still be possible. Particularly, they can be enhanced by also using the low-fidelity models as control variates, not just for bias distribution construction.
            
            Control variate techniques introduce a weighted adjustment term to maintain an unbiased estimator with lower variance. These techniques leverage correlations between information sources to achieve variance reduction. Examples include CV~\cite{rubinstein2016simulation}, approximate CV~\cite{gorodetsky_generalized_2020}, multi-level MC (MLMC)~\cite{giles_2015,giles_multi-level_2020}, multi-index MC (MIMC)~\cite{haji-ali_multi-index_2016}, multi-fidelity MC~\cite{geraci_multifidelity_2017}, and multi-level multi-fidelity (MLMF) MC~\cite{fleeter_multilevel_2020}.
                
                The classical CV technique requires low-fidelity information sources with known mean and known correlation to high-fidelity information sources~\cite{rubinstein2016simulation}, thus making it unsuitable for direct application to the MFUQ problem. Work to extend control variates to the case of unknown covariances was achieved in~\cite{lavenberg_perspective_1981} by leveraging an ensemble estimator that returns an average over a set of CV estimators. However, this estimator still required known means of the low-fidelity information sources.
                
                On the other hand, approximate control variates~\cite{gorodetsky_generalized_2020} were recently developed to tackle the problem of unknown means in the case where the covariance amongst models was known (or easily estimated). This estimator was shown to generalize existing control-variate inspired  techniques such as MFMC and MLMF~\cite{ng_multifidelity_2014,pasupathy_control-variate_2012,peherstorfer_multifidelity_2016, geraci_multifidelity_2017}. Indeed these existing techniques were viewed as either recursive difference or recursive nested estimators that sequentially estimate various unknown means of the low-fidelity sources. Moreover the ACV provides a method to achieve optimal variance reduction in the case of limited high-fidelity resources and increasing low-fidelity resources.
                The ACV can also be viewed as a particular generalization of MLMC and MIMC.  The difference is that these approaches embed recursive difference estimators with a fixed control variate weight ($-1$) within a bias-reduction scheme that sequentially refines the high-fidelity model. Indeed the typical MLMC-type estimators have achieved great success for variance reduction in cases where information sources are related through discretization refinement. Moreover, they do not require estimation of covariances or control variate weights because they use a fixed CV weight.
                
             Nevertheless~\cite{gorodetsky_generalized_2020} showed that the fixed weight of MLMC ($-1$) is only optimal in the case where the correlation between sequential models is one. In all other cases, there exist weights that can further improve variance reduction. However, it has not been clear that this benefit can be truly realized because of the complexity of additional variance introduced by weight estimation from pilot samples. This paper takes the first steps to answer this question by providing guidelines into choosing the number of pilot samples required to achieve reduction.

        Our primary contribution begins to bridge the remaining gap between ACV theory and practice by considering unknown means and unknown correlations (and therefore unknown optimal weights). It does so by adapting the ACV into an ensemble estimator, similar to obtaining pilot samples, and analyzing its performance. These main results are 
            \begin{enumerate}
                \item Theorems~\ref{thm:var_estimators} and \ref{thm:var_cv_sample_weight}: providing the inefficiency induced by unknown correlations as a scaling on the typical $(1-R^2)$ term for ACV-based estimators (ensemble estimator)
                \item Corollary~\ref{thm:lower_bounds} providing for the number of samples required for achieving variance reduction (ensemble estimator)
            \end{enumerate}
            Our derivation of the lower bounds depends on the Gaussianity consumption of the MC estimators and the construction of ensemble estimators which are based on the method of batch means \cite{lavenberg_perspective_1981,nelson_control_1990,nelson_batch_1989}. These bounds are known given prior knowledge on the correlation amongst models, which can be obtained from physics information when possible or lower-bounded to be conservative.
        Two secondary contributions include
            \begin{enumerate}
                \item Theorem~\ref{thm:range_ACV}: explicitly specifies the range of weights over which one still obtains variance reduction, even under errors in approximating the optimal control variate weight
                \item  A new multi-fidelity IS estimator that combines MFIS approach of using the low fidelity model for determining the biasing distribution with a CV estimator that also uses the low fidelity model to leverage its correlations
            \end{enumerate}
        Finally, empirical demonstrations are performed for rare-event estimation with an emphasis on mechanical systems. Multiple examples demonstrate the ensemble estimator theory as well as show improvement over the MFIS estimator.
    
    The rest of this paper is organized as follows. Section~\ref{s:estimators} reviews the MC, IS, and CV estimators along with our target application of rare-event (low-probability-event) simulation. Section~\ref{s:var_cv_e_sample_w} lists our main results on variance reduction via the proposed ensemble-based ACV estimator. The proofs can be found in the appendices. Section~\ref{s:mf_e} introduces how to simply use importance sampling as the underlying estimator within the control variate framework and provides a step-by-step implementation of our estimators in pseudocode. The numerical examples in Section~\ref{s:results} demonstrate our theoretical results by providing empirical performance on synthetic and PDE-based problems. It provides comparison to the MFIS estimator, and comparisons between ensemble and standalone ACV estimators. The paper is concluded in Section~\ref{s:conclusions}.

\section{Background}
\label{s:estimators}

In this section, we review our notation and provide background about the main statistical estimators which are gathered in Table~\ref{tbl:Estimators} for convenience.
\begin{table}[h]
    \centering
    \begin{threeparttable}
    \begin{tabular}{ l l l }
    \toprule
        \bf{Notes} & \textbf{Notation}\tnote{a} & \bf{Equation} \\
    \cmidrule{1-3}
        Monte Carlo & $\mathcal{Q}_n$ & \eqref{e:mc_estimator} \\
    \cmidrule{1-3}
        Importance sampling & $\mathcal{Q}^{\text{IS}}_{q,n}$ & \eqref{e:IS_estimator} \\
    \cmidrule{1-3}
        Classical control variates \\ (known mean and known weight) & $\mathcal{Q}^{\text{CV}}$ & \eqref{e:cv_estimator} \\
     \cmidrule{1-3}
        Approximate control variates \\ (unknown mean and known weight) & $\mathcal{Q}^{e_1}$ & \eqref{e:acv_estimator} \\
     \cmidrule{1-3}
        Ensemble (A)CV-type estimator \\ ( (un)known mean, unknown weight) & $\bar{\mathcal{Q}}^{e_2}(\ubar{\bm{\alpha}}_{e_2})$ &
        \eqref{e:cv_sample_mean} and \eqref{e:acv_sample_mean} \\
    \bottomrule
    \end{tabular}
    \begin{tablenotes}\footnotesize
        \item[a] $e_1 \in \left\{\text{ACV},\text{ACV-X}\right\}$ and $e_2 \in \left\{\text{CV},\text{ACV},\text{ACV-X}\right\}$, where ACV-X indicates a particular sample partitioning scheme used within the ACV.
    \end{tablenotes}
    \caption{List of estimators}
    \label{tbl:Estimators}
    \end{threeparttable}
\end{table}

Let $(\Omega, \mathcal{F}, \prob)$ be a probability space. Let $\mathbb{N}_0$ denote the positive natural numbers, $\reals$ the real number, and $\reals_{+}$ the positive real numbers. For $d \in \mathbb{N}_0$ we use $Z:\Omega \to \reals^d$ to denote a $\mathcal{F}$-measurable  random variable corresponding to input uncertainty. This variable is assumed to be continuous and have probability density function (PDF) $p(\bm{z}):\reals^d \to \reals_{+}$, with $\bm{z} \in \reals^d$ denoting a realization of $Z$. Let $\text{supp}(p)$ be the support of the PDF $p$, i.e., $\text{supp}(p)=\{\bm{z} \in \reals^d, p(\bm{z})>0\}$. Let $Y_i(Z): \reals^d \to \mathcal{R}_i \subset \reals$ for $i=0,1,\ldots,M,$ denote quantities of interest of a high-fidelity model $Y_0$ and $M$ low-fidelity models $\{Y_i\}_{i=1}^M$, respectively. Let the expected values of those models be denoted by $\mu_i=\mathbb{E}_p[Y_i(Z)]$ where $\mathbb{E}_p[\cdot]$ is the expectation operator taken with respect to $p$; for brevity, whenever the PDF $p$ is omitted, it is implicitly assumed, i.e., the expectation $\mathbb{E}[\cdot]$, the variance $\mathbb{V}\text{ar}[\cdot]$ and the covariance $\mathbb{C}\text{ov}[\cdot]$ are computed with respect to $p$. Our goal is to estimate $\mu_0.$%

\subsection{Target application: rare event estimation}
\label{ss:application}
In this section we describe the rare event estimation problem for which we seek to apply variance reduction in Section~\ref{s:results}.
Let $g: \reals^d \to \reals$ denote the limit state function. A failure event is defined by $g(\bm{z})<0$ with the corresponding failure domain
\begin{equation*}
    \mathcal{G}=\left\{ \bm{z} \in \reals^d : g(\bm{z})<0 \right\}.
\end{equation*}
The failure probability $\prob_f(g(\bm{z})<0)$ can be cast as an expectation by considering the indicator function defined on the failure domain $\mathcal{I}_{\mathcal{G}}: \reals^d \to \{0,1\}$ as
\begin{equation*}
    \mathcal{I}_{\mathcal{G}}(\bm{z}) = 
    \begin{cases}
        1, & \bm{z} \in \mathcal{G}, \\
        0, & \text{else}
    \end{cases}.
\end{equation*}
Using this indicator function, the failure probability $\prob_f$ is given by
\begin{equation*}
    \prob_f(g(\bm{z})<0) = \mathbb{E}\left[ \mathcal{I}_{\mathcal{G}}(Z) \right].
\end{equation*}

\subsection{Monte Carlo estimator}
In this section we review the basic Monte Carlo (MC) estimator.
The MC estimator is defined as a normalized sum of random variables
\begin{equation}
  \mathcal{Q}_n(Y_0) = \frac{1}{n}\sum_{i=1}^n Y_0\left(Z^{(i)}\right),
  \label{e:mc_estimator}
\end{equation}
where the random variables $Z^{(i)}$ are independent and identically distributed (i.i.d.) according to the input random variable $Z$, and $\mathcal{Q}_n$ is a new random variable derived from $Y_0$ and each $Z^{(i)}.$

This estimator is unbiased and has variance decaying proportionally to $1/n$ as
\begin{align}
  \mathbb{E}\left[\mathcal{Q}_n(Y_0)\right] &= \mathbb{E}\left[Y_0(Z)\right] = \mu_0, \\
  \vvar{}{\mathcal{Q}_n(Y_0)} &=  \frac{\vvar{}{Y_0(Z)}}{n} = \frac{\sigma_0^2}{n}.
\end{align}
The root mean squared error of the MC estimator is $\sqrt{\vvar{}{\mathcal{Q}_n(Y_0)}} = \sigma_0/\sqrt{n}$, which shows that to gain one more decimal digit of accuracy, the computational cost needs to increase 100 times \cite{mcbook}. In this paper, we seek an estimator with reduced variance by leveraging two ideas: importance sampling~\cite{rubinstein2016simulation} and control variates~\cite{mcbook}.

\subsection{Importance sampling estimator}
\label{ss:is_e}
Importance sampling (IS) seeks variance reduction by carefully choosing a sampling distribution that differs from that of $p$. The IS estimator  $\mathcal{Q}^{\text{IS}}_{q,n}(Y_0)$ is defined by a weighted sum of random variables as
\begin{equation}
  \mathcal{Q}^{\text{IS}}_{q,n}(Y_0) = \frac{1}{n}\sum_{i=1}^n Y_0\left(Z^{(i)}\right) W\left(Z^{(i)}\right),
  \label{e:IS_estimator}
\end{equation}
where the input random variables $Z^{(i)}$ are i.i.d. according to PDF $q(\bm{z})$, and 
\begin{equation}
    W(\bm{z}) = \frac{p(\bm{z})}{q(\bm{z})},
\end{equation}
is the density ratio. The new density $q$ is called the \textit{proposal} (or \textit{biasing}) \textit{density}.

According to \cite{rubinstein2016simulation,mcbook}, the IS estimator is unbiased when $Y_0(\bm{z})p(\bm{z})$ is dominated by $q(\bm{z})$, that is, $q(\bm{z})=0$ implies $Y_0(\bm{z})p(\bm{z})=0$. In other words, if $\text{supp}(Y_0p) \subseteq \text{supp}(q)$, we have
\begin{align}
  \mathbb{E}_{q}\left[\mathcal{Q}^{\text{IS}}_{q,n}(Y_0)\right] &= \mathbb{E}_q\left[Y_0(Z)W(Z)\right] = \mathbb{E}_p\left[Y_0(Z)\right] = \mu_0 \label{e:is_unbias}, \\
  \vvar{q}{\mathcal{Q}^{\text{IS}}_{q,n}(Y_0)} &=  \frac{\vvar{q}{Y_0(Z)W(Z)}}{n}.
\end{align}
A prudent choice of the IS proposal density $q(\bm{z})$ can yield an estimator with a smaller variance than that of the MC estimator. The optimal IS density is obtained by minimizing the variance $\vvar{q}{Y_0(Z)W(Z)}$ as in \cite{rubinstein2016simulation}, and is given in closed form by
\begin{align*}
     q^*(\bm{z}) = \frac{|Y_0(\bm{z})|p(\bm{z})}{\displaystyle\int_{\Omega} |Y_0(\bm{z})|p(\bm{z}) d\bm{z}}.
\end{align*}
Specializing this proposal to the case of rare-event simulation, i.e., $Y_0$ is an indicator function, we obtain
\begin{align}
     q^*(\bm{z}) = \frac{Y_0(\bm{z})p(\bm{z})}{\mu_0}.
     \label{e:IS_density}
\end{align}
Using this proposal, the variance of $\mathcal{Q}^{\text{IS}}_{q,n}(Y_0)$ becomes identically zero because each evaluation of the high-fidelity model $Y_0$ with a single sample drawn from $q^*(\bm{z})$ is exactly equal to its expected value $\mu_0$.

However, since $\mu_0$ is unknown, it is intractable to exactly compute the optimal IS density. Instead, several approaches can be used for obtaining an approximation~\cite{rubinstein2004cross,bugallo_adaptive_2017,papaioannou_sequential_2016}. In this paper, the cross-entropy (CE) method with Gaussian mixture model (GMM) in \citep{geyer_cross_2019,kurtz_cross-entropy-based_2013} is chosen to find an approximate IS density $\hat{q}(\bm{z})$. Here, $\hat{q}(\bm{z})$ is defined as a weighted sum of $k \in \mathbb{N}_0$ multivariate normal density functions 
\begin{equation}
    \hat{q}(\bm{z}) = \sum_{i=1}^k \pi_i \mathcal{N}(\bm{z};\bm{\mu}_i,\bm{\Sigma}_i), \label{e:GMM}
\end{equation}
where $\pi_i \in \mathbb{R}, \bm{\mu}_i \in \mathbb{R}^d,$ and $\bm{\Sigma}_i \in \mathbb{R}^{d \times d}$ for $i=1,2,\ldots,k,$ are the mixture weights, means, and covariance matrices, respectively. The mixture coefficient $\pi_i$ is the probability that the $i^{\text{th}}$ density $\mathcal{N}(\bm{z};\bm{\mu}_i,\bm{\Sigma}_i)$ is selected at a given time \cite{kurtz_cross-entropy-based_2013}, requiring $0 \leq \pi_i \leq 1$ and $\sum_{i=1}^k \pi_i=1$. 

There exist several approaches to estimate the three parameters for each element of the mixture~\cite{kurtz_cross-entropy-based_2013,geyer_cross_2019}. Here we use an expectation-maximization (EM) algorithm with a cross-entropy objective function. We briefly explain this approach in Appendix~\ref{app:EM} by following the construction in \cite{geyer_cross_2019} and refer to \cite{geyer_cross_2019,Chen_emdemystified,rubinstein2016simulation,rubinstein2004cross} for a more detailed treatment.

\subsection{Control variate estimator}
In this section, we first consider the classical control variates---a variance-reduction technique that relies on introducing additional information sources. The classical control variates assume that both the means and the covariances of the additional information sources are available. Next we present an extension of the classical control variates in which the control means are known but the covariances amongst the low-fidelity information sources are unknown \cite{nelson_control_1990,lavenberg_perspective_1981}. Finally, we review the approximate control variates~\cite{gorodetsky_generalized_2020}, which considers the case with with unknown control variate means and known covariances. 

\subsubsection{Classical control variate estimator}
A control variate (CV) estimator $\mathcal{Q}^{\text{CV}}(\bm{\alpha})$ utilizes a set of $M$ additional estimators $\mathcal{Q}_n(Y_1), \ldots, \mathcal{Q}_n(Y_M)$, and augments a baseline estimator $\mathcal{Q}_n(Y_0)$ via a  linear combination of these estimators
\begin{equation}
    \mathcal{Q}^{\text{CV}}(\bm{\alpha}) = \mathcal{Q}_n(Y_0) + \sum_{i=1}^M \alpha_i\left(\mathcal{Q}_n(Y_i) - \mu_i\right) = \mathcal{Q}_n(Y_0) + \bm{\alpha}^{\text{T}}\left(\bm{\mathcal{Q}} - \bm{\mu}\right),
    \label{e:cv_estimator}
\end{equation}
where $\mu_i = \ee{}{\mathcal{Q}_n(Y_i)}$ is the known mean of $\mathcal{Q}_n(Y_i)$, $\bm{\mu} = \left[\mu_1, \ldots, \mu_{M}\right]$, $\bm{\alpha} = [\alpha_1, \ldots, \alpha_M]^{\text{T}}$ is the vector of control variate weights, and $\bm{\mathcal{Q}} = [\mathcal{Q}_n(Y_1), \ldots, \mathcal{Q}_n(Y_M)]^{\text{T}}$ is the vector of additional estimators. These additional estimators are shown here to be Monte Carlo estimators $\mathcal{Q}_n$, but can actually be any random variable. In the later sections we will use importance sampling estimators instead.

This CV estimator is unbiased and has reduced variance compared with the baseline estimator $\mathcal{Q}_n(Y_0)$. Specifically, since $\mathcal{Q}_n(Y_i)$ is unbiased by specification of the estimator above, we have
\begin{equation}
  \ee{}{\mathcal{Q}^{\text{CV}}(\bm{\alpha})} = \mu_0.
\end{equation}
Furthermore, the variance of this estimator is
\begin{align}
  \vvar{}{\mathcal{Q}^{\text{CV}}(\bm{\alpha})} = \vvar{}{\mathcal{Q}_n(Y_0)} + \bm{\alpha}^{\text{T}}\cov{}{\bm{\mathcal{Q}},\bm{\mathcal{Q}}}\bm{\alpha} + 2\bm{\alpha}^{\text{T}} \cov{}{\bm{\mathcal{Q}},\mathcal{Q}_n(Y_0)}.
  \label{e:var_cv_alpha}
\end{align}
The optimal control variate weight \cite{rubinstein2016simulation}, which minimizes the above variance, is then given by
\begin{equation}
  \bm{\alpha}_{\text{CV}}^* = -\bm{C}^{-1}\bm{c},
  \label{e:opt_alpha}
\end{equation}
where $\bm{C} = \cov{}{\bm{\mathcal{Y}},\bm{\mathcal{Y}}} \in \mathbb{R}^{M \times M}$ is the covariance matrix among $Y_i$, $\bm{c} = \cov{}{\bm{\mathcal{Y}},Y_0} \in \mathbb{R}^M$ is the vector of covariances between $Y_0$ and each $Y_i$, and $\bm{\mathcal{Y}} = [Y_1, \ldots, Y_M]^{\text{T}}$. If we further define
\begin{equation}
    \bm{{\bar{c}}} = \bm{c}/\sqrt{\vvar{}{Y_0}} = \left[ \rho_1\sqrt{\vvar{}{Y_1}},\rho_2\sqrt{\vvar{}{Y_2}},\ldots,\rho_M\sqrt{\vvar{}{Y_M}} \right]^{\text{T}},
    \label{e:c_hat}
\end{equation}
where $\rho_i$ is the Pearson correlation coefficient between $Y_0$ and $Y_i$, then the variance corresponding to $\bm{\alpha}_{\text{CV}}^*$ becomes
\begin{equation}
    \vvar{}{\mathcal{Q}^{\text{CV}}(\bm{\alpha}_{\text{CV}}^*)} = \left(1 - R^2\right) \vvar{}{\mathcal{Q}_n(Y_0)},
    \label{e:var_cv}
\end{equation}
where
\begin{equation}
    R^2=\bm{{\bar{c}}}^{\text{T}}\bm{C}^{-1}\bm{{\bar{c}}}.
    \label{e:cv_R^2}
\end{equation}
Here we see that the greater the correlation amongst models, the greater the achieved variance reduction.

Furthermore, since each $\mathcal{Q}_n(Y_i)$ shares $n$ i.i.d. samples, we have
\begin{equation}
    \begin{aligned}
        \cov{}{\mathcal{Q}_n(Y_i),\mathcal{Q}_n(Y_j)} &= 
        \cov{}{\dfrac{1}{n}\displaystyle\sum_{k=1}^n Y_i(Z^{(k)}),\dfrac{1}{n}\displaystyle\sum_{k=1}^n Y_j(Z^{(k)})} = \frac{1}{n} \cov{}{Y_i,Y_j},
    \end{aligned}
\end{equation}
i.e., 
\begin{equation}
    \begin{aligned}
        \bm{C} &= n\bm{\mathsf{C}},  \\
        \bm{c} &= n\bm{\mathsf{c}},
        \label{e:C_c_cov}
    \end{aligned}
\end{equation}
where $\bm{\mathsf{C}} = \cov{}{\bm{\mathcal{Q}},\bm{\mathcal{Q}}}$ and $\bm{\mathsf{c}} = \cov{}{\bm{\mathcal{Q}},\mathcal{Q}_n(Y_0)}$. Thus, we can rewrite~\eqref{e:opt_alpha} as
\begin{equation}
  \bm{\alpha}_{\text{CV}}^* = -\bm{\mathsf{C}}^{-1}\bm{\mathsf{c}}.
  \label{e:opt_alpha_Q}
\end{equation}
to obtain an expression in terms of the variance between estimators.

\subsubsection{Control variates with estimated covariance}
It is often the case that the covariances amongst the low-fidelity information sources are not available; however, these sources can be simulated to obtain estimates of these covariances. One well-analyzed strategy for estimation in this context is to generate an ensemble of $K$ realizations of the random variables $\mathcal{Q}_n(Y_i)$ to estimate the required covariances and correlations \cite{nelson_control_1990,lavenberg_perspective_1981}. Because we are considering MC estimators as additional information sources, this operation requires a total of $nK$ samples for each $Y_i$. The estimated optimal weight is then obtained as
\begin{equation}
    \ubar{\bm{\alpha}}_{\text{CV}} = -\hat{\bm{\mathsf{C}}}^{-1}\hat{\bm{\mathsf{c}}},
    \label{e:cv_sample_alpha}
\end{equation}
where $\hat{\bm{\mathsf{C}}}$ and $\hat{\bm{\mathsf{c}}}$ are the sample versions of $\bm{\mathsf{C}}$ and $\bm{\mathsf{c}}$, respectively. Hence, $\hat{\bm{\mathsf{C}}}$ and $\hat{\bm{\mathsf{c}}}$ can be computed as
\begin{equation}
    \begin{aligned}
        \hat{\bm{\mathsf{C}}} &= \frac{1}{K-1} \sum_{j=1}^K \left( \bm{\mathcal{Q}}_j -  \bar{\bm{\mathcal{Q}}} \right) \left( \bm{\mathcal{Q}}_j -  \bar{\bm{\mathcal{Q}}} \right)^{\text{T}}, \\
        \hat{\bm{\mathsf{c}}} &= \frac{1}{K-1} \sum_{j=1}^K \left( \bm{\mathcal{Q}}_j - \bar{\bm{\mathcal{Q}}} \right)\left(\mathcal{Q}_n^{(j)}(Y_0) - \bar{\mathcal{Q}}_n(Y_0) \right),
    \end{aligned}
    \label{e:cv_sample_covariance}
\end{equation}
where $\bar{\bm{\mathcal{Q}}} = \dfrac{1}{K}\displaystyle\sum_{j=1}^K \bm{\mathcal{Q}}_j$, $\bar{\mathcal{Q}}_n(Y_0) =  \dfrac{1}{K}\displaystyle\sum_{j=1}^K \mathcal{Q}_n^{(j)}(Y_0)$, and $\bm{\mathcal{Q}}_j = \left[Q_n^{(j)}(Y_1), \ldots, Q_n^{(j)}(Y_M)\right]^{\text{T}}$ is the vector of MC estimators using $n$ i.i.d. samples of the $j^{\text{th}}$ batch (out of a total of $K$ batches).

This estimated weight is then used in an ensemble estimator~\cite{lavenberg_perspective_1981}
given by
\begin{equation}
    \begin{aligned}
        \bar{\mathcal{Q}}^{\text{CV}}(\ubar{\bm{\alpha}}_{\text{CV}}) &= \frac{1}{K} \sum_{j=1}^K \mathcal{Q}_j^{\text{CV}}(\ubar{\bm{\alpha}}_{\text{CV}}) = \frac{1}{K} \sum_{j=1}^K \left( \mathcal{Q}_n^{(j)}(Y_0) + \ubar{\bm{\alpha}}_{\text{CV}}^{\text{T}}\left(\bm{\mathcal{Q}}_j - \bm{\mu}_j\right) \right) %
        = \frac{\bm{1}_K^{\text{T}}}{K}\bm{Q}(Y_0) + \ubar{\bm{\alpha}}_{\text{CV}}^{\text{T}}\left(\bar{\bm{\mathcal{Q}}} - \bar{\bm{\mu}}\right),
    \end{aligned}
    \label{e:cv_sample_mean}
\end{equation}
where $\bar{\bm{\mu}} = \dfrac{1}{K}\displaystyle\sum_{j=1}^K \bm{\mu}_j$, $\bm{Q}(Y_0) = \left[ \mathcal{Q}_n^{(1)}(Y_0),\ldots,\mathcal{Q}_n^{(K)}(Y_0) \right]^{\text{T}}$, and $\bm{1}_K$ is a $K \times 1$ vector of ones. 
Comparing~\eqref{e:cv_estimator} and~\eqref{e:cv_sample_mean}, we see that this estimator is identical to the estimator $\mathcal{Q}^{\text{CV}}(\ubar{\bm{\alpha}}_{\text{CV}})$ that uses   the same set of $nK$ samples. Our analytical results in the following sections build from the ensemble form.

\subsubsection{Approximate control variate estimator}
\label{ss:acv}
The approximate control variate (ACV) estimator is designed to leverage the control variate framework when the analytical expectation of the additional estimators $Q_n(Y_i)$ is not known~\cite{gorodetsky_generalized_2020}. This estimator replaces the unknown $\mu_i$ with another estimator $\mu_i^{\text{ACV}}$ as
\begin{align}
    \mathcal{Q}^{\text{ACV}}(\bm{\alpha}) = \mathcal{Q}_n(Y_0) + \sum_{i=1}^M \alpha_i\left(\mathcal{Q}_n(Y_i) - \mu^{\text{ACV}}_i\right) = \mathcal{Q}_n(Y_0) + \bm{\alpha}^{\text{T}}\left(\bm{\mathcal{Q}} - \bm{\mu}^{\text{ACV}}\right),
    \label{e:acv_estimator}
\end{align}
where $\bm{\mu}^{\text{ACV}} = \left[\mu_1^\text{ACV}, \ldots, \mu_{M}^{\text{ACV}}\right]$.
If the $\mu_i^{\text{ACV}}$ are unbiased, then $\mathcal{Q}^{\text{ACV}}(\bm{\alpha})$ is unbiased. Furthermore, the ACV estimator variance is 
\begin{align}
    \vvar{}{\mathcal{Q}^{\text{ACV}}(\bm{\alpha})} = \vvar{}{\mathcal{Q}_n(Y_0)} + \bm{\alpha}^{\text{T}}\cov{}{\bm{\mathcal{Q}}-\bm{\mu}^{\text{ACV}},\bm{\mathcal{Q}}-\bm{\mu}^{\text{ACV}}}\bm{\alpha} + 2\bm{\alpha}^{\text{T}}\cov{}{\bm{\mathcal{Q}}-\bm{\mu}^{\text{ACV}},\mathcal{Q}_n(Y_0)},
    \label{e:var_acv}
\end{align}
so that the optimal weight is 
\begin{align}
    \bm{\alpha}_{\text{ACV}}^* = - \cov{}{\bm{\mathcal{Q}}-\bm{\mu}^{\text{ACV}},\bm{\mathcal{Q}}-\bm{\mu}^{\text{ACV}}}^{-1} \cov{}{\bm{\mathcal{Q}}-\bm{\mu}^{\text{ACV}},\mathcal{Q}_n(Y_0)},
    \label{e:acv_opt_alpha}
\end{align}
corresponding to the variance
\begin{align}
    \vvar{}{\mathcal{Q}^{\text{ACV}}(\bm{\alpha^*}_{\text{ACV}})} = (1-R^2_{\text{ACV}})\vvar{}{\mathcal{Q}_n(Y_0)},
\end{align}
where $R^2_{\text{ACV}} = \cov{}{\bm{\mathcal{Q}}-\bm{\mu}^{\text{ACV}},\mathcal{Q}_n(Y_0)}^{\text{T}} \dfrac{\cov{}{\bm{\mathcal{Q}}-\bm{\mu}^{\text{ACV}},\bm{\mathcal{Q}}-\bm{\mu}^{\text{ACV}}}^{-1}}{\vvar{}{\mathcal{Q}_n(Y_0)}} \cov{}{\bm{\mathcal{Q}}-\bm{\mu}^{\text{ACV}},\mathcal{Q}_n(Y_0)}$.

The ACV estimator is flexible in that it permits a variety of estimators to be used for the unknown means. While there may be many partitioning strategies~\cite{bomarito_optimization_2020}, two specific schemes selected for further analysis in this paper arise from using different partitioning of samples of $Y_i$ for the estimators $\mathcal{Q}_n$ and $\mu_i^{\text{ACV}}.$ Let $\bm{z}_0$, $\bm{z}_i^1$ and $\bm{z}_i^2$ denote the sample sets (realizations of $Z)$ used to compute $\mathcal{Q}_n(Y_0)$, $\mathcal{Q}_n(Y_i)$ and $\mu^{\text{ACV}}_i$, respectively. The two sample partitioning strategies
\cite{gorodetsky_generalized_2020} are defined as
\begin{align}
   \text{ACV-IS} 
   &\begin{cases}
       \bm{z}_i^1 = \bm{z}_0 \\
       \bm{z}_i^2 = \bm{z}_i^1 \cup \tilde{\bm{z}}_i^2 \\
       \tilde{\bm{z}}_i^2 \cap \tilde{\bm{z}}_j^2 = \emptyset \text{ for } i \neq j \\
   \end{cases}, \label{e:ACV-IS} \\
   \text{ACV-MF}
   &\begin{cases}
       \bm{z}_i^1 = \bm{z}_0 \\
       \bm{z}_i^2 = \bm{z}_i^1 \cup \bigcup_{j=1}^i \tilde{\bm{z}}_j \\
       \tilde{\bm{z}}_i \cap \tilde{\bm{z}}_j = \emptyset \text{ for } i \neq j
   \end{cases}, \label{e:ACV-MF}
\end{align}
where $\tilde{\bm{z}}_i^2$ and $\tilde{\bm{z}}_j$ are extra sets of samples to estimate $\mu_i^{\text{ACV}}$. In ACV-IS, the computation of $\mathcal{Q}_n(Y_0)$ and $\mathcal{Q}_n(Y_i)$ employs only the sample set $\bm{z}_0$ while the computation of $\mu^{\text{ACV-IS}}_i$ uses these same samples plus a sample increment. In ACV-MF, besides sharing the sample set $\bm{z}_0$ between $\mathcal{Q}_n(Y_0)$ and $\mathcal{Q}_n(Y_i)$, the estimation of $\mu^{\text{ACV-MF}}_i$ utilizes the sample set of $\mu^{\text{ACV-MF}}_{i-1}$ with some extra samples. We refer to~\cite[Fig. 2]{gorodetsky_generalized_2020} for a visual explanation of the two strategies. No analytical results are available for the optimal sample distribution strategy in the case of finite-sample sizes, and so our aim is to simply show that our analysis applies to a variety of strategies.

According to~\cite{gorodetsky_generalized_2020}, the ACV-IS estimator obtains an optimal weight
\begin{equation}
    \bm{\alpha}^*_{\text{ACV-IS}} = -\left[\bm{C} \circ \bm{F}^{\text{ACV-IS}} \right]^{-1}\left[\text{diag}\left(\bm{F}^{\text{ACV-IS}}\right) \circ \bm{c} \right],
    \label{e:ACV-IS_alpha}
\end{equation}
where $\bm{F}^{\text{ACV-IS}} \in \mathbb{R}^{M \times M}$ has elements
\begin{equation}
    f^{\text{ACV-IS}}_{ij} = \begin{cases}
        \dfrac{r_i-1}{r_i}\dfrac{r_j-1}{r_j} &\mbox{if } i \neq j\\
        \dfrac{r_i-1}{r_i} &\mbox{otherwise}
    \end{cases},
    \label{e:elem_acv_is}
\end{equation}
$\text{diag}(\bullet)$ denotes a vector whose elements are the diagonal of the matrix $\bullet$, and $r_i \in \mathbb{R}_{+}$ is the ratio between the total number of realizations of $Y_i$ and the total number of evaluations of $Y_0$.
Similarly, the ACV-MF estimator obtains an optimal sample weight
\begin{equation}
    \bm{\alpha}^*_{\text{ACV-MF}} = -\left[\bm{C} \circ \bm{F}^{\text{ACV-MF}} \right]^{-1}\left[\text{diag}\left(\bm{F}^{\text{ACV-MF}}\right) \circ \bm{c} \right],
    \label{e:ACV-MF_alpha}
\end{equation}
where $\bm{F}^{\text{ACV-MF}} \in \mathbb{R}^{M \times M}$ has elements
\begin{equation}
    f^{\text{ACV-MF}}_{ij} = \begin{cases}
        \dfrac{\min(r_i,r_j)-1}{\min(r_i,r_j)} &\mbox{if } i \neq j\\
        \dfrac{r_i-1}{r_i} &\mbox{otherwise}
    \end{cases}.
    \label{e:elem_acv_mf}
\end{equation}
The corresponding correlation coefficients are
\begin{align}
    R^2_{\text{ACV-IS}} &= \bm{a}^{\text{T}}\left[\bm{C} \circ \bm{F}^{\text{ACV-IS}}\right]^{-1}\bm{a}, \label{e:R^2_IS} \\
    R^2_{\text{ACV-MF}} &= \bm{b}^{\text{T}}\left[\bm{C} \circ \bm{F}^{\text{ACV-MF}}\right]^{-1}\bm{b}, \label{e:R^2_MF}
\end{align}
where
\begin{equation*}
    \begin{aligned}
        \bm{a} &= \text{diag}\left(\bm{F}^{\text{ACV-IS}}\right) \circ \bar{\bm{c}}, \\
        \bm{b} &= \text{diag}\left(\bm{F}^{\text{ACV-MF}}\right) \circ \bar{\bm{c}}.
    \end{aligned}
\end{equation*}
We can rewrite the ACV optimal weight in terms of the covariances amongst MC estimators $\bm{\mathsf{C}}$ and $\bm{\mathsf{c}}$ (rather than amongst models $Y_i$) by substituting~\eqref{e:C_c_cov} into~\eqref{e:ACV-IS_alpha} and~\eqref{e:ACV-MF_alpha} as
\begin{align}
    \bm{\alpha}^*_{\text{ACV-IS}} &= -\left[\bm{\mathsf{C}} \circ \bm{F}^{\text{ACV-IS}} \right]^{-1}\left[\text{diag}\left(\bm{F}^{\text{ACV-IS}}\right) \circ \bm{\mathsf{c}} \right], \label{e:alpha_star_is} \\
   \bm{\alpha}^*_{\text{ACV-MF}} &= -\left[\bm{\mathsf{C}} \circ \bm{F}^{\text{ACV-MF}} \right]^{-1}\left[\text{diag}\left(\bm{F}^{\text{ACV-MF}}\right) \circ \bm{\mathsf{c}} \right]. \label{e:alpha_star_mf}
\end{align}

\subsubsection{Approximate control variates with estimated covariance}
In this section we describe an extension to the ACV estimator that considers unknown covariances amongst the low-fidelity information sources. This extension utilizes the same idea of~\cite{lavenberg_perspective_1981} to construct a new estimator as an average of an ensemble of estimators.

Suppose again that the optimal weight is estimated from an ensemble of $K$ simulations, each of which employs the same number of samples. Using the sample weight $\ubar{\bm{\alpha}}_{\text{ACV}}$, 
we define the ensemble ACV estimator as
\begin{align}
        \bar{\mathcal{Q}}^{\text{ACV}}(\ubar{\bm{\alpha}}_{\text{ACV}}) &= \frac{1}{K} \sum_{j=1}^K \mathcal{Q}_j^{\text{ACV}}(\ubar{\bm{\alpha}}_{\text{ACV}}) = \frac{1}{K} \sum_{j=1}^K \left( \mathcal{Q}_n^{(j)}(Y_0) + \ubar{\bm{\alpha}}_{\text{ACV}}^{\text{T}}\left(\bm{\mathcal{Q}}_j - \bm{\mu}^{\text{ACV}}_j\right) \right) \nonumber \\
        &= \frac{\bm{1}_K^{\text{T}}}{K}\bm{Q}(Y_0) + \left(\bar{\bm{\mathcal{Q}}} - \bar{\bm{\mu}}^{\text{ACV}}\right)^{\text{T}}\ubar{\bm{\alpha}}_{\text{ACV}},
    \label{e:acv_sample_mean}
\end{align}
where $\bar{\bm{\mu}}^{\text{ACV}} = \dfrac{1}{K} \displaystyle\sum_{j=1}^K \bm{\mu}^{\text{ACV}}_j$.

The next section analyzes the variance reduction of the proposed ensemble estimators $\bar{Q}^{\text{ACV}}$, both for MF and IS sampling strategies.

\section{Ensemble ACV estimators}
\label{s:var_cv_e_sample_w}

The goal of this section is to present Theorem~\ref{thm:var_estimators}, which expresses the variances of the three ensemble estimators derived from the CV and ACV  (i.e., $\bar{\mathcal{Q}}^{\text{CV}}(\ubar{\bm{\alpha}}_{\text{CV}})$, $\bar{\mathcal{Q}}^{\text{ACV-IS}}(\ubar{\bm{\alpha}}_{\text{ACV-IS}})$ and $\bar{\mathcal{Q}}^{\text{ACV-MF}}(\ubar{\bm{\alpha}}_{\text{ACV-MF}})$) with respect to the number of samples, the correlation coefficient, the variance of the MC estimator $\mathcal{Q}_n(Y_0)$, and the expectation of a function of estimators. We use this relationship in Theorem~\ref{thm:var_cv_sample_weight} to compute explicit expressions for these variances. In Corollary~\ref{thm:lower_bounds} we derive lower bounds on the number of ensembles required to guarantee smaller variances of the three sample-weight estimators than that of the baseline estimator $\mathcal{Q}_n(Y_0)$.

The results rely on a multivariate Gaussianity assumption, and so are asymptotically true for model settings where the information sources have finite mean and variance as an implication from the central limit theorem. That is, the results holds when the vectors $\{\mathcal{Q}_n(Y_0),\mathcal{Q}_n(Y_1),\ldots,\mathcal{Q}_n(Y_M)\}$ in the CV-based estimator and \[\{\mathcal{Q}_n(Y_0),\mathcal{Q}_n(Y_1),\ldots,\mathcal{Q}_n(Y_M),\mathcal{Q}_{nr_1}(Y_1),\ldots,\mathcal{Q}_{nr_M}(Y_M)\}\] in the ACV-based estimators have a multivariate normal distribution as $n \to \infty.$

Using the Gaussianity assumption of the vectors of MC estimators, the first theorem allows us to calculate the variances of the ensemble estimators in terms of the expectations shown below.
\begin{theorem}
    \label{thm:var_estimators}
    \renewcommand{\theenumi}{\alph{enumi}}
    ~\begin{enumerate}
        \item Let the vector $\{\mathcal{Q}_n(Y_0),\mathcal{Q}_n(Y_1),\ldots,\mathcal{Q}_n(Y_M)\}$ have a multivariate normal distribution. Then we have
        \begin{align}
            \vvar{}{\bar{\mathcal{Q}}^{\text{CV}}(\ubar{\bm{\alpha}}_{\text{CV}})}
            &= \vvar{}{\mathcal{Q}_n(Y_0)}\left(1-R^2\right) \left( \frac{1}{K} + \ee{\bm{\mathcal{Q}}}{\left(\bar{\bm{\mathcal{Q}}} - \bar{\bm{\mu}}^{\text{CV}}\right)^{\text{T}} \left( \bm{D}\bm{D}^{\text{T}}\right)^{-1} \left(\bar{\bm{\mathcal{Q}}} - \bar{\bm{\mu}}^{\text{CV}}\right)} \right),
            \label{e:var_q_cv}
        \end{align}
        where $R$ is defined in~\eqref{e:cv_R^2}.
        \item Let the vector $\{\mathcal{Q}_n(Y_0),\mathcal{Q}_n(Y_1),\ldots,\mathcal{Q}_n(Y_M),\mu_1^e,\ldots,\mu_M^e\}$, where $e \in \{\text{ACV-IS},\text{ACV-MF}\}$ and $\mu_i^e = \mathcal{Q}_{nr_i}(Y_i)$ for $i=1,2,\ldots,M$, have a multivariate normal distribution. Then, we have
        \begin{equation}
        \begin{aligned}
            \vvar{}{\bar{\mathcal{Q}}^e(\ubar{\bm{\alpha}}_e)} = \;&\vvar{}{\mathcal{Q}_n(Y_0)}\left(1-R^2_e\right) \\ 
            &\times \left( \frac{1}{K} + \ee{\tilde{\bm{\mathcal{Q}}}^e}{\left(\bar{\bm{\mathcal{Q}}} - \bar{\bm{\mu}}^e\right)^{\text{T}} \left[\left(\bm{D}\bm{D}^{\text{T}}\right) \circ \bm{F}^e \right]^{-1}
            \left[ \left(\bm{D}\bm{D}^{\text{T}}\right) \circ \bm{\mathcal{F}}^e_M \circ \left(\bm{\mathcal{F}}^e_M\right)^{\text{T}} \right]
            \left[\left(\bm{D}\bm{D}^{\text{T}}\right) \circ \bm{F}^e \right]^{-1}
            \left(\bar{\bm{\mathcal{Q}}} - \bar{\bm{\mu}}^e\right)} \right)
            \label{e:var_q_acv}
        \end{aligned}
        \end{equation}
        where $\tilde{\bm{\mathcal{Q}}}^e = \left[\mathcal{Q}_n(Y_1)-\mu_1^e,\ldots,\mathcal{Q}_n(Y_M)-\mu_M^e\right]^{\text{T}}$, $\bm{\mathcal{F}}^e_M = \text{diag}(\bm{F}^e) \otimes \bm{1}_M$, and $R^2_e$ is defined in~\eqref{e:R^2_IS} and~\eqref{e:R^2_MF}.
    \end{enumerate}
\end{theorem}
The proof of Theorem~\ref{thm:var_estimators} in Appendix~\ref{app:var_estimators} requires the identities provided in Appendix~\ref{app:identities} and the following two useful propositions. In the first proposition, we rewrite the sample weight in terms of the centered data matrix to facilitate the calculation of the variances of the ensemble estimators.
\begin{proposition}
    \label{pro:alpha_sample_cov}
    The estimated control variate weights can be written as
    \begin{align}
        \ubar{\bm{\alpha}}_{\text{CV}} &= -\left( \bm{D}\bm{D}^{\text{T}} \right)^{-1} \bm{D} \bm{Q}(Y_0), \label{e:alpha_cv_sample} \\
        \ubar{\bm{\alpha}}_e &= -\left[\left(\bm{D}\bm{D}^{\text{T}}\right) \circ \bm{F}^e \right]^{-1}\left[\text{diag}\left(\bm{F}^e\right) \circ \left(\bm{D} \bm{Q}(Y_0)\right) \right], \label{e:alpha_acv_sample}
    \end{align}
    where $e \in \{\text{ACV-IS},\text{ACV-MF}\}$ and $\bm{D}$ is the centered data matrix
    \begin{equation}
        \bm{D} = \begin{bmatrix}
            \mathcal{Q}_n^{(1)}(Y_1) - \bar{\mathcal{Q}}_n(Y_1) & \mathcal{Q}_n^{(2)}(Y_1) - \bar{\mathcal{Q}}_n(Y_1) & \ldots & \mathcal{Q}_n^{(K)}(Y_1) - \bar{\mathcal{Q}}_n(Y_1) \\
            \mathcal{Q}_n^{(1)}(Y_2) - \bar{\mathcal{Q}}_n(Y_2) & \mathcal{Q}_n^{(2)}(Y_2) - \bar{\mathcal{Q}}_n(Y_2) & \ldots & \mathcal{Q}_n^{(K)}(Y_2) - \bar{\mathcal{Q}}_n(Y_2) \\
            \ldots & \ldots & \ldots & \ldots \\
            \mathcal{Q}_n^{(1)}(Y_M) - \bar{\mathcal{Q}}_n(Y_M) & \mathcal{Q}_n^{(2)}(Y_M) - \bar{\mathcal{Q}}_n(Y_M) & \ldots & \mathcal{Q}_n^{(K)}(Y_M) - \bar{\mathcal{Q}}_n(Y_M)
            \label{e:centered_data_matrix}
        \end{bmatrix}
    \end{equation}
\end{proposition}
\begin{proof}
See Appendix~\ref{app:alpha_acv_sample}.
\end{proof}
The next step is to calculate the expectations in Theorem~\ref{thm:var_estimators} by finding the distribution of the expressions inside these operators. Coupled with the Gaussianity assumption, the specific structure of these expressions suggests the Hotelling's $T^2$ distribution~\cite{nist2012handbook}, which is confirmed by the second proposition.
\begin{proposition}
    \label{pro:q_mu_dist}
    \renewcommand{\theenumi}{\alph{enumi}}
    ~\begin{enumerate}
        \item Let the vector $\{\mathcal{Q}_n(Y_0),\mathcal{Q}_n(Y_1),\ldots,\mathcal{Q}_n(Y_M)\}$ have a multivariate normal distribution. Then,
        \begin{align}
            \bar{\bm{\mathcal{Q}}} - \bar{\bm{\mu}}^{\text{CV}} \sim \mathcal{N}\left( \bm{0}_M, \frac{\bm{\mathsf{C}}}{K} \right)
            \label{e:q_mu_cv_dist}
        \end{align}
        \item Let the vector $\{\mathcal{Q}_n(Y_0),\mathcal{Q}_n(Y_1),\ldots,\mathcal{Q}_n(Y_M),\mu_1^e,\mu_2^e,\ldots,\mu_M^e\}$, where $e \in \{\text{ACV-IS},\text{ACV-MF}\}$ and $\mu_i^e = \mathcal{Q}_{nr_i}(Y_i)$ for $i=1,2,\ldots,M$, have a multivariate normal distribution. Then,
        \begin{align}
            \bar{\bm{\mathcal{Q}}} - \bar{\bm{\mu}}^{e} &\sim \mathcal{N}\left(\bm{0}_{M}, \dfrac{\bm{\mathsf{C}} \circ \bm{F}^{e}}{K} \right)
            \label{e:q_mu_acv_dist}
        \end{align}
    \end{enumerate}
\end{proposition}
\begin{proof}
See Appendix~\ref{app:q_mu_dist}.
\end{proof}

We can obtain explicit expressions for the expectation in Theorem~\ref{thm:var_estimators} under certain reasonable limiting conditions on the ratio of low-fidelity to high-fidelity samples $r_i$. Theorem~\ref{thm:var_cv_sample_weight} summarize the variance reduction ratios of the ensemble CV-type estimators with respect to the baseline estimator $\mathcal{Q}_n(Y_0)$.
\begin{theorem}
    \label{thm:var_cv_sample_weight}
    \renewcommand{\theenumi}{\alph{enumi}}
    ~\begin{enumerate}
        \item Let the vector $\{\mathcal{Q}_n(Y_0),\mathcal{Q}_n(Y_1),\ldots,\mathcal{Q}_n(Y_M)\}$ have a multivariate normal distribution. Then,
        \begin{equation}
            \frac{\vvar{}{\mathcal{Q}^{\text{CV}}(\ubar{\bm{\alpha}}_{\text{CV}})}}{\vvar{}{\mathcal{Q}_n(Y_0)}} = (1-R^2) \left( 1 + \frac{M}{K-M-2} \right)
            \label{e:var_ratio_cv}
        \end{equation}
        \item Let the vector $\{\mathcal{Q}_n(Y_0),\mathcal{Q}_n(Y_1),\ldots,\mathcal{Q}_n(Y_M),\mu_1^e,\ldots,\mu_M^e\}$, where $e \in \{\text{ACV-IS},\text{ACV-MF}\}$ and $\mu_i^e = \mathcal{Q}_{nr_i}(Y_i)$ for $i=1,2,\ldots,M$, have a multivariate normal distribution. If we further assume that
        \begin{align}
            \text{(ACV-IS)} &\quad r_i \gg 1, \label{e:assm_acv_is} \\
            \text{(ACV-MF)} &\quad r_i = r,
            \label{e:assm_acv_mf}
        \end{align}
        for $i=1,2,\ldots,M$, then
        \begin{align}
            \frac{\vvar{}{\mathcal{Q}^e(\ubar{\bm{\alpha}}_e)}}{\vvar{}{\mathcal{Q}_n(Y_0)}} &= (1-R^2_e) \left( 1 + \frac{a(e)M}{K-M-2} \right),
            \label{e:var_ratio_acv}
        \end{align}
        where $a(\text{ACV-IS}) = 1$ and $a(\text{ACV-MF}) = \dfrac{r-1}{r}$.
    \end{enumerate}
\end{theorem}
\begin{proof}
See Appendix~\ref{app:var_cv_sample_weight}.
\end{proof}

To guarantee variance reduction, the number of ensembles $K$ must be bounded from below as follows
\begin{corollary}
    Suppose $R^2 \neq 0$, $R^2_{e_1} \neq 0$ for $e_1 \in \{\text{ACV-IS},\text{ACV-MF}\}$. Furthermore, let $e_2 \in \{\text{CV},\text{ACV-IS},\text{ACV-MF}\}$ and
    \begin{equation}
        K > \max (M+2, B_{e_2} ),
        \label{e:K_B}
    \end{equation}
    where
    \begin{align}
        B_{\text{CV}} &= \frac{M}{R^2}+2, \\
        B_{\text{ACV-IS}} &= \frac{M}{R^2_{\text{ACV-IS}}}+2, \\
        B_{\text{ACV-MF}} &= \frac{r-1}{r}\frac{M}{R^2_{\text{ACV-MF}}} + \frac{M}{r} + 2.
        \label{e:B}
    \end{align}
    Then, we have
    \begin{equation*}
        \frac{\vvar{}{\mathcal{Q}^{e_2}(\ubar{\bm{\alpha}}_{e_2})}}{\vvar{}{\mathcal{Q}_n(Y_0)}} < 1.
    \end{equation*}
    \label{thm:lower_bounds}
\end{corollary}
\begin{proof}
See Appendix~\ref{app:num_samples}.
\end{proof}

The above corollary implies that in order for $\vvar{}{\mathcal{Q}^e(\ubar{\bm{\alpha}}_e)} < \vvar{}{\mathcal{Q}_n(Y_0)}$, i.e., variance reduction of the CV-type estimators with estimated weight compared to the baseline estimator, the number of realizations $K$ of the random variable $\mathcal{Q}^e(\ubar{\bm{\alpha}}_e)$ has to be at least $\max (M+2,B_e)$, where $B_e$ depends on the sample partitioning, the correlation amongst models, and the number of low-fidelity models. This corollary recovers the result in \cite{lavenberg_perspective_1981} for the CV estimator.

From the proof of the corollary in Appendix~\ref{app:num_samples}, the results are obtained by setting the upper bound of the ratio $\dfrac{\vvar{}{\mathcal{Q}^e(\ubar{\bm{\alpha}}_e)}}{\vvar{}{\mathcal{Q}_n(Y_0)}}$ (i.e., $y$ in \eqref{e:num_samples}) equal to 1, which satisfies the constraint $y+R^2_e > 1$ for \textit{any} non-zero $R_e$, i.e., when $y=1$, we only need the models to be correlated and do not need a specific value of the correlation coefficient. In case we know the correlation amongst models, we may choose a smaller upper bound such that $y > 1 - R^2_e$, e.g., if $R^2_e=0.9$, $y$ can take any value in the interval $(0.1, 1]$. This reflects the principle of the CV-based methods: stronger correlation leads to smaller variance.

The corollary explicitly ties the correlation amongst models to the sampling requirements. This type of connection has previously been ignored in the general case of the ACV-like estimators. Furthermore, it provides an avenue through which to inject \textit{problem specific} information as correlations. In this respect it can be used in multi-level Monte Carlo~\cite{giles_multi-level_2020} schemes for which a convergence rate for a numerical method is used to determine optimal allocations and guarantee convergence. We envision that these types of problems can also be amenable to deriving expressions for the correlation.

Finally, to choose $K$ and the numbers of samples in practice we can solve an optimization problem in which (1) $K$ and the numbers of samples are design variables; (2) the variance of an appropriate estimator is minimized; and (3) the total cost is bounded above and $K$ is constrained by~\eqref{e:K_B}. We leave solving such a problem for future work.

\section{Multi-fidelity importance sampling control-variate estimator}
\label{s:mf_e}

We now combine the CV-type estimators with importance sampling (IS) to target rare-event calculations. The resulting estimator will closely parallel that of~\cite{peherstorfer_multifidelity_2016}, with the primary difference being the leveraging of control variates for further variance reduction.

Specifically, instead of the Monte Carlo estimator, we now use importance sampling as the baseline estimators for both the CV and ACV estimators, yielding
\begin{equation}
    \mathcal{Q}^{\text{MF}}_{\hat{q},n}\left( \mathcal{Q}^{\text{IS}}_{\hat{q},n}(Y_0), \mathcal{Q}^{\text{IS}}_{\hat{q},n}(Y_1), \alpha\right) =  \mathcal{Q}^{\text{IS}}_{\hat{q},n}(Y_0) + \alpha\left( \mathcal{Q}^{\text{IS}}_{\hat{q},n}(Y_1) - \mu_1\right)
    \label{e:MF}
\end{equation}
for the CV and 
\begin{equation}
    \mathcal{Q}^{\text{MF-ACV}}_{\hat{q},n}\left( \mathcal{Q}^{\text{IS}}_{\hat{q},n}(Y_0), \mathcal{Q}^{\text{IS}}_{\hat{q},n}(Y_1), \alpha\right) =  \mathcal{Q}^{\text{IS}}_{\hat{q},n}(Y_0) + \alpha\left( \mathcal{Q}^{\text{IS}}_{\hat{q},n}(Y_1) - \mu_1^{\text{IS}}\right) \label{e:MF_ACV}
\end{equation}
for the ACV, where $\mu_1$ is the known mean of $Y_1$ and $\mu_1^{\text{IS}}$ is the IS estimator for the mean of $Y_1$, i.e., $\mu_1^{\text{IS}} = \mathcal{Q}^{\text{IS}}_{\hat{q},m}(Y_1)$ for $m \in \mathbb{N}_0$. Here, $\hat{q}$ is a the biasing distribution that has the property $\text{supp}(Y_0p) \subseteq \text{supp}(\hat{q})$. For simplicity of presentation, we only consider the case with a single additional low fidelity model.
Now Eqs. \eqref{e:var_cv_alpha}, \eqref{e:var_acv}, \eqref{e:opt_alpha} and \eqref{e:acv_opt_alpha} become
\begin{align}
    \vvar{\hat{q}}{\mathcal{Q}^{\text{MF}}_{\hat{q},n}\left( \mathcal{Q}^{\text{IS}}_{\hat{q},n}(Y_0), \mathcal{Q}^{\text{IS}}_{\hat{q},n}(Y_1), \alpha\right)} &= \vvar{\hat{q}}{\mathcal{Q}^{\text{IS}}_{\hat{q},n}(Y_0)} + \alpha^2 \vvar{\hat{q}}{\mathcal{Q}^{\text{IS}}_{\hat{q},n}(Y_1)} + 2\alpha \cov{\hat{q}}{\mathcal{Q}^{\text{IS}}_{\hat{q},n}(Y_0),\mathcal{Q}^{\text{IS}}_{\hat{q},n}(Y_1)}, \label{e:var_MF} \\
    \vvar{\hat{q}}{\mathcal{Q}^{\text{MF-ACV}}_{\hat{q},n}\left( \mathcal{Q}^{\text{IS}}_{\hat{q},n}(Y_0), \mathcal{Q}^{\text{IS}}_{\hat{q},n}(Y_1), \alpha\right)} &= \vvar{\hat{q}}{\mathcal{Q}^{\text{IS}}_{\hat{q},n}(Y_0)} + \alpha^2\Big( \vvar{\hat{q}}{\mathcal{Q}^{\text{IS}}_{\hat{q},n}(Y_1)} + \vvar{\hat{q}}{\mu_1^{\text{IS}}} - 2\cov{\hat{q}}{\mathcal{Q}^{\text{IS}}_{\hat{q},n}(Y_1),\mu_1^{\text{IS}}} \Big) \nonumber \\
    &\;\;\; + 2\alpha
    \left(\cov{\hat{q}}{\mathcal{Q}^{\text{IS}}_{\hat{q},n}(Y_0),\mathcal{Q}^{\text{IS}}_{\hat{q},n}(Y_1)} - \cov{\hat{q}}{\mathcal{Q}^{\text{IS}}_{\hat{q},n}(Y_0), \mu_1^{\text{IS}}}\right), \label{e:var_MF_ACV} \\
    \alpha^*_{\text{CV}} &= - \frac{\cov{\hat{q}}{\mathcal{Q}^{\text{IS}}_{\hat{q},n}(Y_0),\mathcal{Q}^{\text{IS}}_{\hat{q},n}(Y_1)}}{\vvar{\hat{q}}{\mathcal{Q}^{\text{IS}}_{\hat{q},n}(Y_1)}}, \\
    \alpha^*_{\text{ACV}} &= - \frac{\cov{\hat{q}}{\mathcal{Q}^{\text{IS}}_{\hat{q},n}(Y_0),\mathcal{Q}^{\text{IS}}_{\hat{q},n}(Y_1)} - \cov{\hat{q}}{\mathcal{Q}^{\text{IS}}_{\hat{q},n}(Y_0), \mu_1^{\text{IS}}}}{\vvar{\hat{q}}{\mathcal{Q}^{\text{IS}}_{\hat{q},n}(Y_1)} + \vvar{\hat{q}}{\mu_1^{\text{IS}}} - 2\cov{\hat{q}}{\mathcal{Q}^{\text{IS}}_{\hat{q},n}(Y_1), \mu_1^{\text{IS}}}}. \label{e:opt_alpha_ACV_MF}
\end{align}

\subsection{Properties of the MF estimators}
The proposed estimators are unbiased.

\begin{theorem}
    Suppose $\text{supp}(Y_0p) \subseteq \text{supp}(\hat{q})$. Then, $\mathcal{Q}^{\text{MF}}_{\hat{q},n}$ and $\mathcal{Q}^{\text{MF-ACV}}_{\hat{q},n}$ are unbiased estimators of the expected value $\mu_0$.
    \label{thm:unbias}
\end{theorem}
\begin{proof}
Let us consider the expected values of $\mathcal{Q}^{\text{MF}}_{\hat{q},n}$ and $\mathcal{Q}^{\text{MF-ACV}}_{\hat{q},n}$ with respect to $\hat{q}$
\begin{equation*}
    \begin{aligned}
        \ee{\hat{q}}{\mathcal{Q}^{\text{MF}}_{\hat{q},n}(\mathcal{Q}^{\text{IS}}_{\hat{q},n}(Y_0),\mathcal{Q}^{\text{IS}}_{\hat{q},n}(Y_1), \alpha)} &= \ee{\hat{q}}{\mathcal{Q}^{\text{IS}}_{\hat{q},n}(Y_0)} + \alpha\ee{\hat{q}}{\mathcal{Q}^{\text{IS}}_{\hat{q},n}(Y_1) - \mu_1}, \\
        \ee{\hat{q}}{\mathcal{Q}^{\text{MF-ACV}}_{\hat{q},n}\left( \mathcal{Q}^{\text{IS}}_{\hat{q},n}(Y_0), \mathcal{Q}^{\text{IS}}_{\hat{q},n}(Y_1), \alpha\right)} &= \ee{\hat{q}}{\mathcal{Q}^{\text{IS}}_{\hat{q},n}(Y_0)} + \alpha\ee{\hat{q}}{\mathcal{Q}^{\text{IS}}_{\hat{q},n}(Y_1) - \mu_1^{\text{IS}}}.
    \end{aligned}
\end{equation*}
Since $\text{supp}(Y_0p) \subseteq \text{supp}(\hat{q})$, which allows us to apply \eqref{e:is_unbias}, we have $\ee{\hat{q}}{\mathcal{Q}^{\text{IS}}_{\hat{q},n}(Y_i)} = \ee{}{Y_i(Z)} = \mu_i \text{ for } i=0,1,$ and $\ee{\hat{q}}{\mu_1^{\text{IS}}} = \ee{}{Y_1(Z)} = \mu_1$. Thus,
\begin{equation*}
    \begin{aligned}
        \ee{\hat{q}}{\mathcal{Q}^{\text{MF}}_{\hat{q},n}(\mathcal{Q}^{\text{IS}}_{\hat{q},n}(Y_0),\mathcal{Q}^{\text{IS}}_{\hat{q},n}(Y_1), \alpha)} &= \mu_0, \\
        \ee{\hat{q}}{\mathcal{Q}^{\text{MF-ACV}}_{\hat{q},n}\left( \mathcal{Q}^{\text{IS}}_{\hat{q},n}(Y_0), \mathcal{Q}^{\text{IS}}_{\hat{q},n}(Y_1), \alpha\right)} &= \mu_0,
    \end{aligned}
\end{equation*}
which implies $\mathcal{Q}^{\text{MF}}_{\hat{q},n}$ and $\mathcal{Q}^{\text{MF-ACV}}_{\hat{q},n}$ are unbiased estimators of the expected value $\mu_0$.
\end{proof}

This result means that the Gaussian mixture can be used as a proposal.
\begin{corollary}
    If $\hat{q}$ is a Gaussian mixture as in \eqref{e:GMM}, $\mathcal{Q}^{\text{MF}}_{\hat{q},n}$ and $\mathcal{Q}^{\text{MF-ACV}}_{\hat{q},n}$ are unbiased estimators of the expected value $\mu_0$.
\end{corollary}
\begin{proof}
Because $\pi_i$ are probabilities and $\mathcal{N}(\bm{z};\bm{\mu}_i,\bm{\Sigma}_i) > 0,\forall \bm{z} \in \mathbb{R}^d,i=1,2,\ldots,k$, the GMM in \eqref{e:GMM} has global support, that is, $\hat{q}(\bm{z})>0, \forall \bm{z} \in \mathbb{R}^d$. This leads to $\text{supp}(Y_0p) \subseteq \text{supp}(\hat{q})$ and, hence, Theorem~\ref{thm:unbias} holds.
\end{proof}

As with any control-variate estimator, variance reduction is greater when the correlation between estimators $\mathcal{Q}^{\text{IS}}_{\hat{q},n}(Y_0)$ and $\mathcal{Q}^{\text{IS}}_{\hat{q},n}(Y_1)$ is larger. In fact, as shown in the below theorems, for \textit{any} non-zero correlation the variance of our estimator is smaller than that of the multi-fidelity importance sampling (MFIS) estimator $\mathcal{Q}^{\text{MFIS}}_{\hat{q},n}$ presented in \cite{peherstorfer_multifidelity_2016}. Specifically, the MFIS estimator uses a proposal density that is derived from $Y_1$. In other words, the MFIS estimator leverages low-fidelity information sources to design the proposal distribution and then uses this proposal distribution within a standard importance sampling scheme for the high-fidelity model $Y_0.$ 

Assuming that we use the same proposal, our proposed estimator is guaranteed to have lower variance than the MFIS for a certain range of the control weight. This further reduction is achieved by leveraging the low-fidelity models again as control variates. 
When the control variate weight is zero and the proposal distribution matches, we obtain an equivalent estimator
$\mathcal{Q}^{\text{MFIS}}_{\hat{q},n}\left( Y_0 \right) =  \mathcal{Q}^{\text{IS}}_{\hat{q},n}(Y_0).
$, when $\alpha=0.$

We can generally choose the weight $\alpha$ to achieve greater variance reduction by leveraging the correlation between $Y_0$ and $Y_1$.
Formally, Theorem~\ref{thm:CV} states that if the weight $\alpha$ belongs to a certain range and $\vvar{}{Q(Y_1)} \neq 0$, then the variance of the CV estimator $\mathcal{Q}^{\text{CV}}(\alpha)$ is bounded above by that of the estimator $Q(Y_0)$. The equality happens when there is no correlation between $Q(Y_0)$ and $Q(Y_1)$.

\begin{theorem}[Range of control variate weight for CV estimator] \label{thm:CV}
    Let  $\mathcal{Q}^{\text{CV}}(\alpha) = Q(Y_0) + \alpha \left(Q(Y_1) - \mu_1\right)$ with $\vvar{}{Q(Y_1)} > 0$, where $Q(Y_0)$ and $Q(Y_1)$ are unbiased estimators for the means of $Y_0$ and $Y_1$, respectively. Furthermore, let 
    \begin{equation}
        \bar{f} = -\dfrac{2\cov{}{Q(Y_0),Q(Y_1)}}{\vvar{}{Q(Y_1)}} 
    \end{equation}
    denote a scaled ratio of the covariance to the variance of $Q(Y_1)$. 
    
    If $\bar{f} \geq 0$ (resp. $\bar{f} \leq 0$), then for control variate weight in the range $\alpha \in \left[0, \bar{f}\right]$ (resp. $\alpha \in \left[\bar{f}, 0\right]$), the variance of the CV estimator is bounded above by that of the baseline estimator, i.e., $\vvar{}{\mathcal{Q}^{\text{CV}}(\alpha)} \leq \vvar{}{Q(Y_0)}.$
    
    Equality is obtained for $\alpha = 0.$
\end{theorem}
\begin{proof}
Let
\begin{align*}
    V(\alpha) &= \vvar{}{\mathcal{Q}^{\text{CV}}(\alpha)} - \vvar{}{Q(Y_0)} \\
      &= \left( \vvar{}{Q(Y_0)} + \alpha^2\vvar{}{Q(Y_1)} + 2\alpha\cov{}{Q(Y_0),Q(Y_1)}  \right) - \vvar{}{Q(Y_0)} \\
      &= \alpha\left(\alpha\vvar{}{Q(Y_1)} + 2\cov{}{Q(Y_0),Q(Y_1)}\right),
\end{align*}
where the second equality uses~\eqref{e:var_cv_alpha} with one low-fidelity model.
Therefore,
\begin{align*}
    V(\alpha) \leq 0 \iff
    \left[\begin{aligned}
        &\alpha \geq 0 \text{ and } \alpha \leq \bar{f} \Rightarrow \text{ if } \bar{f} \geq 0, \text{ then } \alpha \in \left[0, \bar{f}\right]. \\
        &\alpha \leq 0 \text{ and } \alpha \geq \bar{f} \Rightarrow \text{ if } \bar{f} \leq 0, \text{ then } \alpha \in \left[\bar{f}, 0\right].
    \end{aligned}
    \right.
\end{align*}
\end{proof}

The below corollary confirms the advantage of our estimator over the MFIS one: aside from the trivial case of independent models, using an appropriate value of $\alpha$ ensures variance reduction.
\begin{corollary}[Range of control variate weight for MF estimator]
    Let 
    \begin{align*}
        \tilde{f} &= -\dfrac{2\cov{\hat{q}}{\mathcal{Q}^{\text{IS}}_{\hat{q},n}(Y_0),\mathcal{Q}^{\text{IS}}_{\hat{q},n}(Y_1)}}{\vvar{\hat{q}}{\mathcal{Q}^{\text{IS}}_{\hat{q},n}(Y_1)}}, 
    \end{align*}
    and $\vvar{\hat{q}}{\mathcal{Q}^{\text{IS}}_{\hat{q},n}(Y_1)} > 0$. If $\tilde{f} \geq 0$ (resp. $\tilde{f} \leq 0$), then for control variate weight in the range $\alpha \in [0,\tilde{f}]$ (resp. $\alpha \in [\tilde{f},0]$) the variance of the MF estimator is bounded above by that of the MFIS estimator, i.e., $\vvar{\hat{q}}{\mathcal{Q}^{\text{MF}}_{\hat{q},n}\left( \mathcal{Q}^{\text{IS}}_{\hat{q},n}(Y_0), \mathcal{Q}^{\text{IS}}_{\hat{q},n}(Y_1), \alpha\right)} \leq \vvar{\hat{q}}{\mathcal{Q}^{\text{MFIS}}_{\hat{q},n}(Y_0)}$.
    
    Equality is obtained for $\alpha = 0.$
    \label{cor:MF_MFIS}
\end{corollary}
\begin{proof}
Substitute $p, Q(Y_0)$ and $Q(Y_1)$ in Theorem~\ref{thm:CV} with $\hat{q}, \mathcal{Q}^{\text{IS}}_{\hat{q},n}(Y_0)$ and $\mathcal{Q}^{\text{IS}}_{\hat{q},n}(Y_1)$, respectively.
\end{proof}

Theorem~\ref{thm:CV} can be extended straightforwardly to the approximate control variates described in Section~\ref{ss:acv}.
\begin{theorem}[Range of control variate weight for ACV estimator]
    \label{thm:range_ACV}
    Let
    \begin{align*}
         s_1 &= \vvar{}{Q(Y_1) - \mu_1^{\text{ACV}}} > 0, \\
         s_2 &= \cov{}{Q(Y_0),Q(Y_1)-\mu_1^{\text{ACV}}}.
    \end{align*}
    If $s_2 \leq 0$ (resp. $s_2 \geq 0$), then for control variate weight in the range $\alpha \in \left[0,-\dfrac{2s_2}{s_1}\right]$ $\left(\textit{resp. } \alpha \in \left[-\dfrac{2s_2}{s_1},0\right]\right)$ the variance of the ACV estimator is bounded above by that of the baseline estimator, i.e., $\vvar{}{\mathcal{Q}^{\text{ACV}}(\alpha)} \leq \vvar{}{Q(Y_0)}$.
    
    Equality is obtained for $\alpha = 0.$
\end{theorem}
\begin{proof}
This theorem is proved by following the same steps as in the proof of Theorem~\ref{thm:CV}, which are first simplify the difference $V(\alpha) = \vvar{}{\mathcal{Q}^{\text{ACV}}(\alpha)} - \vvar{}{Q(Y_0)}$ using~\eqref{e:var_acv} with one low-fidelity model, and then find $\alpha$ such that $V(\alpha) \leq 0$. We skip the details for brevity.
\end{proof}

In words, the theorem states that if the variance of $Q(Y_1) - \mu_1^{\text{ACV}}$ is non-zero, we can always choose $\alpha$ from an interval depending on the sign of $\cov{}{Q(Y_0),Q(Y_1)} - \cov{}{Q(Y_0), \mu_1^{\text{ACV}}}$ so that the ACV estimator has smaller variance than the MC estimator.

These results and the ensemble estimator variance reduction results provide strong motivation towards using the control-variate weight for further variance reduction even when the weight cannot be exactly estimated. These results indicate that there is a range of values of the weight that still leads to variance reduction, i.e., that some weight estimate can have an error but still be beneficial. The ensemble estimator variance results provide a sufficient condition for the number of samples required to guarantee a particular size of variance reduction.

\subsection{Algorithms}

Algorithms \ref{al:MF_CCV} and \ref{al:MF_ACV} provide pseudocode to implement the ensemble estimators for the control variate and the approximate control variate with the importance sampling approach. These algorithms requires the selection of a low-fidelity model correlated to the high-fidelity model, the input parameters to the EM algorithm, and the number of samples. Using these inputs, the algorithms produce an estimate of the expected value of the high-fidelity model. We note that Algorithm \ref{al:MF_ACV} uses the ACV-IS strategy, and adapting it to the ACV-MF strategy is trivial and omitted for brevity.

\begin{algorithm}[htb]
    \caption{Ensemble control variate estimator using importance sampling}
    \label{al:MF_CCV}
    \begin{algorithmic}[1]
    \REQUIRE $Y_0,Y_1$: high-fidelity and low-fidelity model; $p$: PDF of the input random variables; $\mu_1$: expected value of $Y_1$, i.e., $\mu_1=\ee{}{Y_1}$; $n_s,\tau,k_{\text{init}}$: parameters of the EM algorithm; $R$: correlation coefficient; $\mathcal{C}$: target cost, or $\zeta$: target accuracy;
    \ENSURE $\hat{\mu}_0$: estimate of the expected value of $Y_0$, i.e., $\hat{\mu}_0\approx\ee{}{Y_0}$; $v_{\text{CV}}$: estimate of the sample variance of the ensemble CV estimator.
    \STATE Compute the approximate IS density $\hat{q}$ using the EM algorithm with the parameters $n_s,\tau,$ and $k_{\text{init}}$
    \STATE Determine the number of outer loops $K$ and the number of samples $n$ using $R$ and $\mathcal{C}$ (or $\zeta$) (e.g., solving an optimization problem) 
    \FOR{$i=1,2,\ldots,K$}
        \STATE Draw $n$ samples $\bm{z}_1,\bm{z}_2,\ldots,\bm{z}_n$ from the density $\hat{q}$
    \FOR{$j=1,2,\ldots,n$}
        \STATE $Q_{j}(Y_k) = Y_k(\bm{z}_j)\widehat{W}(\bm{z}_j)$ for $k=0,1$
    \ENDFOR
        \STATE $\bar{Q}_i(Y_k) = \dfrac{1}{n}\displaystyle\sum_{j=1}^{n} Q_{j}(Y_k)$ for $k=0,1$
    \ENDFOR
    \STATE $\tilde{Q}(Y_k) = \dfrac{1}{K}\displaystyle\sum_{i=1}^{K}\bar{Q}_i(Y_k)$ for $k=0,1$
    \STATE $\hat{c} = \dfrac{1}{K-1} \displaystyle\sum_{i=1}^K \left( \bar{Q}_i(Y_0) - \tilde{Q}(Y_0) \right)\left( \bar{Q}_i(Y_1) - \tilde{Q}(Y_1) \right)$; $\hat{s}_k = \dfrac{1}{K-1} \displaystyle\sum_{i=1}^K \left( \bar{Q}_i(Y_k) - \tilde{Q}(Y_k) \right)^2$ for $k=0,1$
    \STATE $\ubar{\alpha} = -\dfrac{\hat{c}}{\hat{s}_1}$
    \STATE $\hat{\mu}_0 = \tilde{Q}(Y_0) + \ubar{\alpha}\left( \tilde{Q}(Y_1) - \mu_1 \right)$; $\bar{v}_{\text{CV}} = \dfrac{1}{K} \left( \hat{s}_0 + \ubar{\alpha}^2 \hat{s}_1 + 2\ubar{\alpha} \hat{c} \right)$ \hfill\COMMENT{\eqref{e:var_MF}}
    \end{algorithmic}
\end{algorithm}

\begin{algorithm}[htb]
    \caption{Ensemble approximate control variate estimator using importance sampling}
    \label{al:MF_ACV}
    \begin{algorithmic}[1]
    \REQUIRE $Y_0,Y_1$: high-fidelity and low-fidelity model; $p$: PDF of the input random variables; $n_s,\tau,k_{\text{init}}$: parameters of the EM algorithm; $R$: correlation coefficient; $\mathcal{C}$: target cost, or $\zeta$: target accuracy;
    \ENSURE $\hat{\mu}_0$: estimate of the expected value of $Y_0$, i.e., $\hat{\mu}_0\approx\ee{}{Y_0}$; $v_{\text{IS}}$: estimate of the sample variance of the ensemble ACV-IS estimator.
    \STATE Compute the approximate IS density $\hat{q}$ using the EM algorithm with the parameters $n_s,\tau,$ and $k_{\text{init}}$
    \STATE Determine the number of outer loops $K$, and the number of samples $n$ and $m$ using $R$ and $\mathcal{C}$ (or $\zeta$) (e.g., solving an optimization problem) 
    \FOR{$i=1,2,\ldots,K$}
        \STATE Draw $n$ samples $\{\bm{z}_1,\bm{z}_2,\ldots,\bm{z}_{n}\}$ and $m$ samples $\{\bm{z}'_1,\bm{z}'_2,\ldots,\bm{z}'_{m}\}$ from the density $\hat{q}$
    \FOR{$j=1,2,\ldots,n$}
        \STATE $Q_{j}(Y_k) = Y_k(\bm{z}_j)\widehat{W}(\bm{z}_j)$ for $k=0,1$
    \ENDFOR
    \FOR{$j=1,2,\ldots,m$}
        \STATE $Q'_{j}(Y_1) = Y_1(\bm{z}'_j)\widehat{W}(\bm{z}'_j)$
    \ENDFOR
    \STATE $\bar{Q}_i(Y_k) = \dfrac{1}{n}\displaystyle\sum_{j=1}^{n} Q_{j}(Y_k)$ for $k=0,1$
    \STATE $\bar{Q}'_i(Y_1) = \dfrac{1}{n+m}\left( \displaystyle\sum_{j=1}^{n} Q_{j}(Y_1) + \displaystyle\sum_{j=1}^{m} Q'_{j}(Y_1) \right)$
    \ENDFOR
    \STATE $\tilde{Q}(Y_k) = \dfrac{1}{K}\displaystyle\sum_{i=1}^{K}\bar{Q}_i(Y_k)$ for $k=0,1$; $\tilde{Q}'(Y_1) = \dfrac{1}{K}\displaystyle\sum_{i=1}^{K}\bar{Q}'_i(Y_1)$
    \STATE $\hat{s}_k = \dfrac{1}{K-1} \displaystyle \sum_{i=1}^{K} \left( \bar{Q}_i(Y_k) - \tilde{Q}(Y_k) \right)^2$ for $k=0,1$; $\hat{s}' = \dfrac{1}{K-1} \displaystyle \sum_{i=1}^{K} \left( \bar{Q}'_i(Y_1) - \tilde{Q}'(Y_1) \right)^2$
    \STATE $\hat{c} = \dfrac{1}{K-1} \displaystyle \sum_{i=1}^{K} \left( \bar{Q}_i(Y_0) - \tilde{Q}(Y_0) \right) \left( \bar{Q}_i(Y_1) - \tilde{Q}(Y_1) \right)$
    \STATE $\hat{c}'_k = \dfrac{1}{K-1} \displaystyle \sum_{i=1}^{K} \left( \bar{Q}_i(Y_k) - \tilde{Q}(Y_k) \right) \left( \bar{Q}'_i(Y_1) - \tilde{Q}'(Y_1) \right)$ for $k=0,1$
    \STATE $\ubar{\alpha} = -\dfrac{\hat{c}-\hat{c}'_0}{\hat{s}_1 + \hat{s}' - 2\hat{c}'_1}$ \hfill\COMMENT{\eqref{e:opt_alpha_ACV_MF}}
    \STATE $\hat{\mu}_0 = \tilde{Q}(Y_0) + \ubar{\alpha} \left( \tilde{Q}(Y_1) - \tilde{Q}'(Y_1) \right)$; $\bar{v}_{\text{IS}} = \dfrac{1}{K} \left( \hat{s}_0 + \ubar{\alpha}^2 (\hat{s}_1 + \hat{s}' - 2\hat{c}'_1) + 2\ubar{\alpha} (\hat{c}-\hat{c}'_0) \right)$ \hfill\COMMENT{\eqref{e:var_MF_ACV}}
    \end{algorithmic}
\end{algorithm}

\section{Numerical Results}
\label{s:results}
In this section we demonstrate the performance of our ensemble importance-sampling control variate algorithms for rare-event estimation for three examples. The first example is a simple case of estimating tail probabilities involving normal random variables. The second example is a cantilever beam whose material uncertainty is modeled as a random field discretized by the Karhunen-Lo\`{e}ve expansion. Our third example analyzes a clamped Mindlin plate in bending under random loads and material properties. 
To focus on demonstrating the benefit of the control variates and exercising our theory, we focus only on problems with a single low-fidelity model.

The simplicity of the first example allows cheap evaluations of its models, and so enables us to disambiguate between sources of errors, such as lack of optimality in the proposal distribution. The second example is more realistic than the first one in the sense that it solves a system of PDEs by the finite element method and computing output statistics on a dense mesh is often expensive. Based on the similar settings as the first example of \cite{peherstorfer_multifidelity_2016}, our last example aims to stress the benefit of our proposed estimator on a more complex problem.

For each example we provide the governing equations, definitions of limit state functions, and the choice of HF and LF models. We then describe implementation details such as parameters of the EM algorithm and number of samples. In each example we aim to compare the variance of several estimator to compare their performance with the MFIS estimator under equal costs. In each example we fix the total cost to $\mathcal{C}_{\text{MFIS}}$, and then only use a number of high fidelity samples $n_{\text{HF}}$ and low-fidelity samples $n_{\text{LF}}$ to ensure
$\mathcal{C}_{\text{MFIS}} = n_{\text{HF}}\mathcal{C}_{\text{HF}} + n_{\text{LF}}\mathcal{C}_{\text{LF}}$, where  $\mathcal{C}_{\text{HF}}$ and $\mathcal{C}_{\text{LF}}$ are the cost of the MFIS estimator, the cost of one evaluation of the HF and LF models, respectively. Note that we only compare \textit{online} cost, where the low-fidelity distribution was already computed, because the offline cost is equal for all algorithms (i.e., they all use the same biasing distribution). The values of $n_{\text{HF}}$ and $n_{\text{LF}}$ are determined in each example from either a stated assumption or from the runtime of the implementation.

All examples~\cite{cvis_codes} are implemented in Matlab and use the cross-entropy (CE) code from \cite{CE_codes} with some trivial modifications to integrate with our estimator. All quantities given below are dimensionless for simplicity. For convenience, the figures in this section use shortened notations for the variances of considered estimators, i.e.,
\begin{align}
    v_{\text{CV}} &= \vvar{\hat{q}}{\mathcal{Q}^{\text{MF}}_{\hat{q},n}}, \\
    v_0 &= \vvar{\hat{q}}{\mathcal{Q}^{\text{MFIS}}_{\hat{q},n}},  \\
    v_{\text{IS}} &= \vvar{\hat{q}}{\mathcal{Q}^{\text{MF-1}}_{\hat{q},n}}, \\
    \bar{v}_{\text{CV}} &= \vvar{\hat{q}}{\bar{\mathcal{Q}}^{\text{MF}}_{\hat{q},K}}, \\
    \bar{v}_{\text{IS}} &= \vvar{\hat{q}}{\bar{\mathcal{Q}}^{\text{MF-1}}_{\hat{q},K}},
\end{align}
where $\bar{\mathcal{Q}}^{\text{MF}}_{\hat{q},K}$ and $\bar{\mathcal{Q}}^{\text{MF-1}}_{\hat{q},K}$ are the ensemble estimators defined in \eqref{e:cv_sample_mean} and \eqref{e:acv_sample_mean} when replacing $\mathcal{Q}^{\text{CV}}$ and $\mathcal{Q}^{\text{ACV}}$ with $\mathcal{Q}^{\text{MF}}_{\hat{q},n}$ and $\mathcal{Q}^{\text{MF-1}}_{\hat{q},n}$, respectively, and using $K$ ensembles; and MF-1 indicates the use of the ACV-IS scheme. Estimates of the variance of each ensemble estimators are available from Algorithms \ref{al:MF_CCV} and \ref{al:MF_ACV}, and for reference we recall the expressions of $v_{\text{CV}}$ and $v_{\text{IS}}$ are Eqs.~\eqref{e:var_MF} and~\eqref{e:var_MF_ACV}. However, these variances are only valid for cases with large $K$. In cases with small $K$ these algorithms need to run many times to evaluate the empirical variances of the estimators. Lastly, we note that the true mean in the control variate approach is determined from either analytical expression (the first example) or using an extra set of a very large number of samples (the second and third example).

\subsection{Analytical example}
\label{ss:ex-1}

We first consider an analytical example where we seek to evaluate tail probabilities of Gaussians. Our aim is to explore: (1) how non-optimality in the proposal distribution affects the variance reduction, and (2) how much benefit we obtain in both over the process that uses a single-fidelity importance sampling estimator based on a low-fidelity proposal~\cite{peherstorfer_multifidelity_2016}.

We consider a standard normal input space $Z \sim \mathcal{N}(0,1)$, and two limit-state functions $g_0$ and $g_1$. In this problem we will treat $g_0$ as the high-fidelity model. The failure probabilities can be analytically computed using the standard normal cumulative distribution function $\Psi$ according to $P_f(g_0(z) < 0) = 1 - \Psi(l_0)$, where $l_0=3$ is our chosen threshold. The high-fidelity limit-state function becomes $g_0(z) = l_0 - z.$

A low-fidelity model would have an error in the failure probability threshold, and we posit that such an error occurs from an incorrect specification of a threshold. In this case we use $l_1 = 2.8$ as the threshold to make this model ``lower-fidelity'' $g_1(z) = l_1 - z.$

\subsubsection{Experiment with threshold sequence}
\begin{figure}[h]
    \centering
    \includegraphics[width=2.5in]{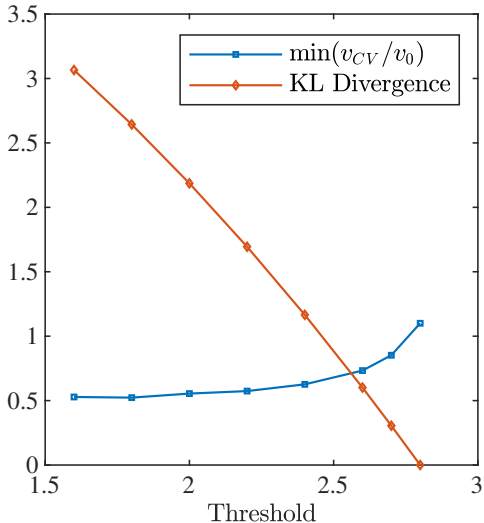}
    \caption{Minimum variance ratio of $v_{\text{CV}}$ to $v_0$, and KL divergence between the approximate low-fidelity distribution and the exact low-fidelity distribution vs. threshold $l^{(i)}$}
    \label{f:Ex1:tol_min}
\end{figure}
The cross-entropy method does not provide a fine-grained control over approximation error due to the approximation quality of the Gaussian mixture model. However, we would like to investigate how such approximation quality affects algorithm performance. To this end, this subsection investigates an alternative way to generate proposal distributions with fine-grained control on error.

Our alternative approach is to use a sequence of biasing distributions obtained by specifying an alternate set of limits $g^{(i)}(z) = l^{(i)} - z$, $l^{(i)} \in \{1.6,1.8,\dots,2.8\}$, and to use rejection sampling to sample form them exactly. When $l^{(i)} = 2.8$, we are exactly sampling from the optimal proposal for the low-fidelity model. As the threshold decreases to $l^{(i)}=1.6$, our proposal distribution has increasing error (as any GMM proposal would). In other words, as the intermediate thresholds $l^{(i)}$ approach the LF threshold $l_1$, the corresponding intermediate IS densities progressively become better and ultimately the low-fidelity IS density. In this manner, we have disambiguated the error due to sub-optimal biasing distributions and those due to low-fidelity effects.

Figure~\ref{f:Ex1:tol_min} shows two curves: the red curve is the  Kullback-Leibler (KL) divergence between the approximate low-fidelity distribution and the exact low-fidelity distribution, and the blue curve is the minimum variance ratio $\min(v_{\text{CV}}/v_0)$ between the control-variate and the MFIS estimators.
As the KL divergence converges to zero, the variance ratio $\min(v_{\text{CV}}/v_0)$ approaches 1. This behavior is expected because it corresponds to perfect sampling of the low-fidelity model, which leads to a zero variance estimate for the low-fidelity model. Specifically, as this variance vanishes, the contribution to the covariance vanishes, and we recover $v_{\text{CV}}=v_0$.

We also see that as the KL divergence increases, $\min(v_{\text{CV}}/v_0)$ reaches a plateau. This behavior aligns well with the fact that the variance reduction must depend on the correlation amongst models and we cannot reduce $v_{\text{CV}}/v_0$ to 0 simply by using more crude approximations of the LF biasing distribution.

\subsubsection{Varying the control variate weight for a fixed proposal}
Next we show that even if the estimate of the control variate weight is not extremely accurate in practice, there is still a interval of weights in which our estimator is able to achieve a larger variance reduction.

Figure~\ref{f:Ex-1:tol_15} shows the ratio of the variance between $v_{\text{CV}}$ and $v_0$ with respect to varying weight $\alpha$ for the proposal based on $l^{(i)}=1.6.$ Such a dependence is quadratic according to~\eqref{e:var_MF}.
The threshold of equal performance is shown in red.
It is clear from Figure~\ref{f:Ex-1:tol_15} that there is a range of $\alpha$ in which the CV estimator has small variance than the MFIS estimator. This fact has been established in Corollary~\ref{cor:MF_MFIS}. In other words, although the CV estimator has an additional parameter to estimate---there is a range of weights for which it still improves upon the baseline estimator.

\begin{figure}
    \centering
    \includegraphics[width=2.5in]{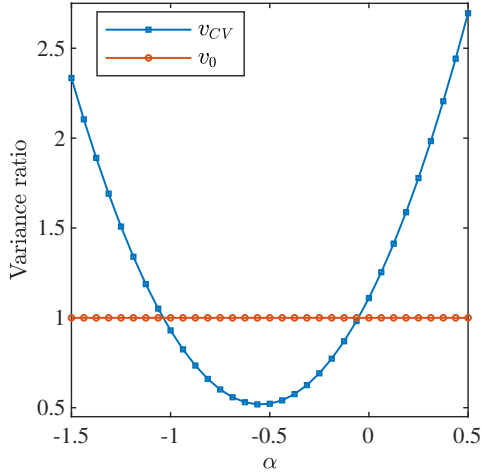}
    \caption{$\alpha$ vs. $v_{\text{CV}}/v_0$ for $l^{(i)}=1.6$.}
    \label{f:Ex-1:tol_15}
\end{figure}
\begin{figure}[h]
    \centering
    \includegraphics[width=3in]{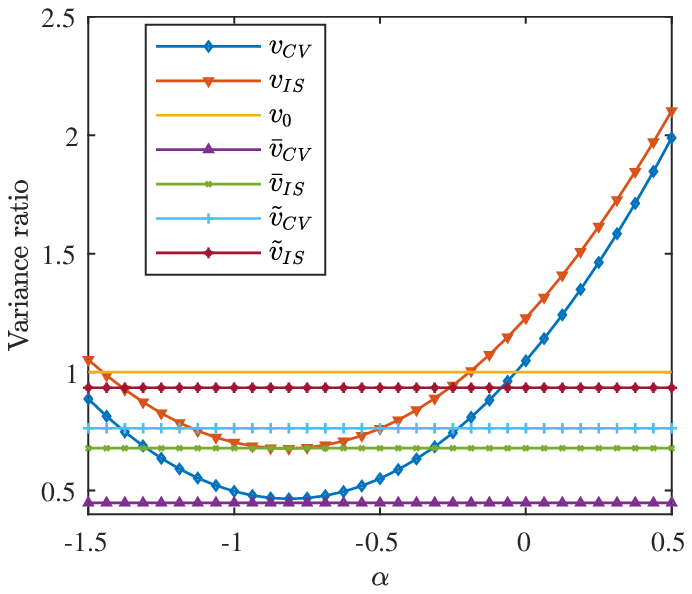}
    \caption{$\alpha$ vs. variance ratios using the EM algorithm with GMM; $\tilde{v}_{\text{CV}} = \vvar{\hat{q}}{\bar{\mathcal{Q}}^{\text{MF}}_{\hat{q},K}}$ and
    $\tilde{v}_{\text{IS}} = \vvar{\hat{q}}{\bar{\mathcal{Q}}^{\text{MF-1}}_{\hat{q},K}}$ with small $K$ (i.e., $K=4$).}
    \label{f:Ex-1:ce_ve_v0_alpha}
\end{figure}

\subsubsection{Experiment with a fixed threshold using the EM algorithm}
We now switch from the accurate biasing distribution sampling to the EM algorithm approach of Section~\ref{ss:is_e}. This algorithm is deployed to construct the approximate density from the low-fidelity model (i.e., $l_1=2.8$), from which samples are drawn to calculate the estimated variances
in Figure~\ref{f:Ex-1:ce_ve_v0_alpha}. The input parameters of the EM algorithm are $n_s=3000, \tau=0.1$ and the initial number of mixture components $k_{\text{init}}=3$. The variance of the baseline estimator $v_0$ is calculated using $n_0=5 \times 10^5$ samples. We assume further that for this example the HF model is 30 times more expensive to evaluate than the LF model.
As shown in Table~\ref{tbl:ex1_sample}, the number of samples used to calculate other estimators
is chosen such that they consume the same cost as the baseline estimator. To compute the variances of the ensemble estimators, i.e., $\bar{v}_{\text{CV}}$ and $\bar{v}_{\text{IS}}$, we utilize the sample sets of the corresponding component estimators (i.e., $\mathcal{Q}^{\text{MF}}_{\hat{q},n}$ and $\mathcal{Q}^{\text{MF-1}}_{\hat{q},n}$) but divide each of them into $K=1000$ batches, and then apply Algorithms~\ref{al:MF_CCV} and~\ref{al:MF_ACV}. 
\begin{table}[h]
    \centering
    \begin{threeparttable}
    \begin{tabular}{ c c c c }
    \toprule
        & $v_0$ & $v_{\text{CV}}$ & $v_{\text{IS}}$ \\
    \cmidrule{1-4}
        $n_{\text{HF}}$ & 500000 & 483870 & 434782 \\
    \cmidrule{1-4}
        $n_{\text{LF}}$ & & 483870 & 1956519 \\
    \bottomrule
    \end{tabular}
    \caption{Sample allocation for the first example.}
    \label{tbl:ex1_sample}
    \end{threeparttable}
\end{table}

Figure~\ref{f:Ex-1:ce_ve_v0_alpha} shows that:
\begin{enumerate}[label=\arabic*., topsep=0pt, parsep=0pt, partopsep=0pt, itemsep=0pt]
    \item The ACV estimators  $\mathcal{Q}^{\text{MF}}_{\hat{q},n}$ and $\mathcal{Q}^{\text{MF-1}}_{\hat{q},n}$ have smaller variances than that of MFIS over a certain range of weights;
    \item As reported in~\cite{gorodetsky_generalized_2020}, $v_{\text{CV}}$ is the smallest among the estimators since it uses the exact mean of $Y_1$;
    \item Since $\bar{v}_{\text{CV}}$ and $\bar{v}_{\text{IS}}$ utilize estimated optimal control weight, their variances approximate the minimum of $v_{\text{CV}}$ and $v_{\text{IS}}$;
    \item Even with small $K$, the ensemble estimators $\tilde{v}_{\text{CV}}$ and $\tilde{v}_{\text{IS}}$ are still smaller than $v_0$ according to Corollary~\ref{thm:lower_bounds}.
\end{enumerate}
As a reference, Figure~\ref{f:Ex1:ce_qce} shows the HF, LF and approximate density---denoted by $q_0$, $q_1$ and $\hat{q}$, respectively. Here, $\hat{q}$ is the approximate optimal density $q_1$ obtained  via the GMM approximation.
Being the cross-entropy approximation to the LF density, the function $\hat{q}$ fluctuates around $q_1$; in particular, the vertical line of $q_1$ at $2.8$ is estimated by a very steep curve.
\begin{figure}[ht]
    \centering
    \includegraphics[width=2.5in]{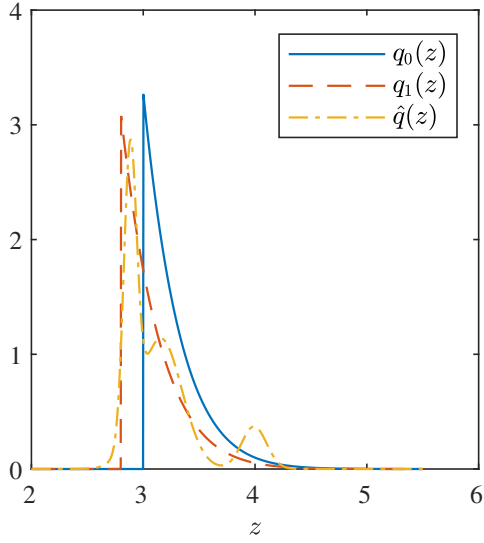}
    \caption{High-fidelity, low-fidelity and approximate density.}
    \label{f:Ex1:ce_qce}
\end{figure}

\subsection{Cantilever beam}
\label{ss:ex-2}

In this section, we compare estimators for a cantilever beam with uncertain material properties. %

\begin{figure}[ht]
    \centering
    \includegraphics[scale=0.3]{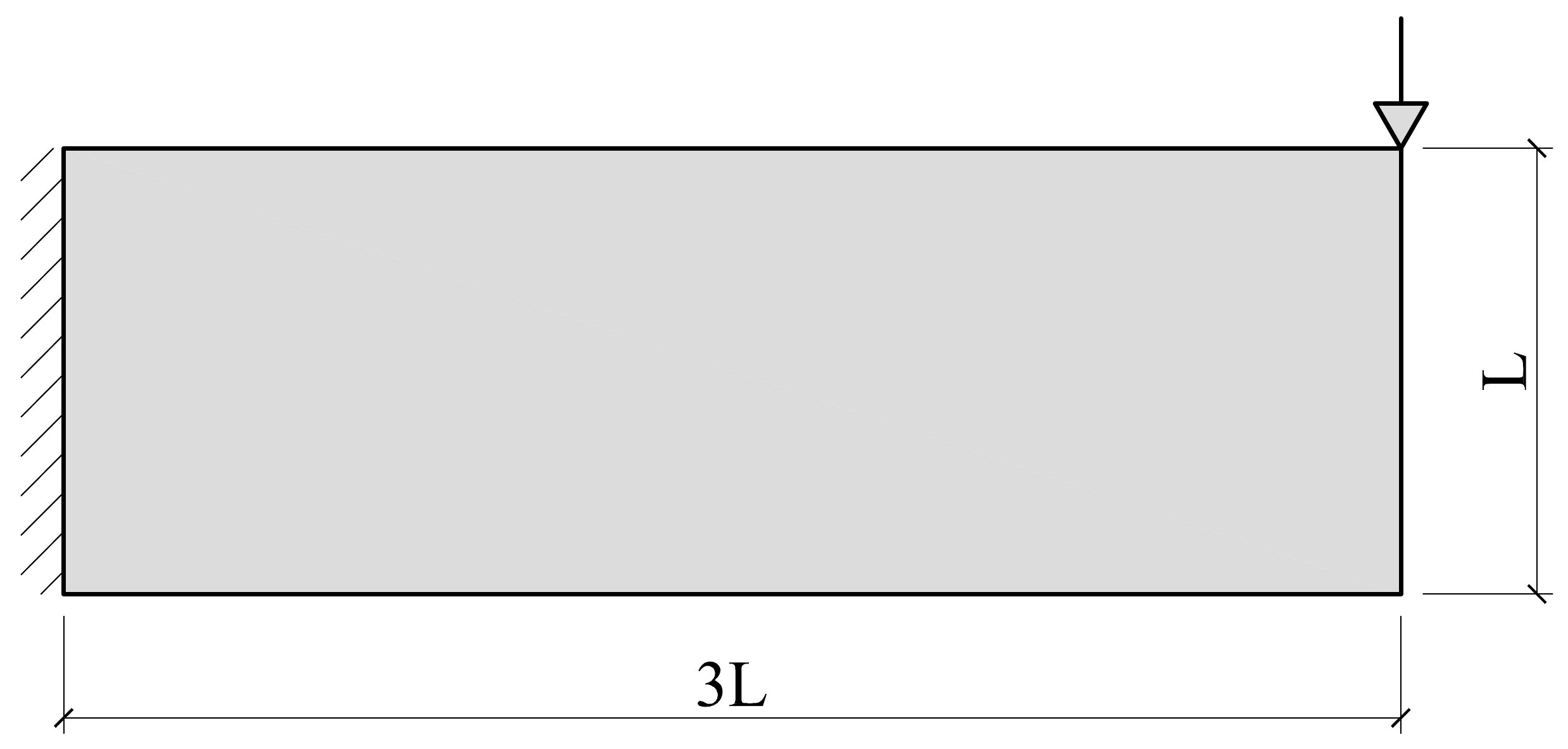} 
    \caption{The cantilever beam.}
    \label{fig:C}
\end{figure}
Figure~\ref{fig:C} displays the cantilever beam, a two-dimensional domain used in this example. The beam is fixed on its left side, subject to an vertical load at its upper-right corner, and has the dimensions of $3L \times L$ unit length. Let $\mathcal{X}=[0,3L]\times[0,L]\subset\mathbb{R}^2$ denote the domain in Figure~\ref{fig:C}. Considering a rectangular Cartesian coordinate system and the Einstein summation convention, the governing equations \cite{sadd2020elasticity} of a two-dimensional, infinitesimal strain, linear elastic problem without body force are given as
\begin{align}
    \sigma_{ij,j} &= 0, \nonumber \\
    \sigma_{ij} &= \lambda\epsilon_{kk}\delta_{ij} + 2\mu\epsilon_{ij}, \label{e:stress-strain} \\
    \epsilon_{ij} &= \frac{1}{2}(u_{i,j}+u_{j,i}), \nonumber \\
    u_i(0,x_2) &= 0, \label{e:boundary}
\end{align}
where the $(\bullet)_{,j}$ subscript stands for $\partial (\bullet)/\partial x_j$; $\sigma_{ij}$ is the Cauchy stress tensor, $u_i$ the displacement, $\epsilon_{ij}$ the strain, and $\delta_{ij}$ the Kronecker delta. The boundary condition \eqref{e:boundary} reflects the fixed left side of the beam. In \eqref{e:stress-strain} $\lambda$ and $\mu$ are the Lam\'{e} constants which can be calculated from the Young's modulus $E$ and Poisson's ratio $\nu$ as follows
\begin{align*}
    \lambda = \frac{E\nu}{(1+\nu)(1-2\nu)},\, \mu = \frac{E}{2(1+\nu)}.
\end{align*}
\begin{figure}[htb]
     \centering
     \begin{subfigure}[b]{0.495\textwidth}
         \centering
         \includegraphics[width=3in]{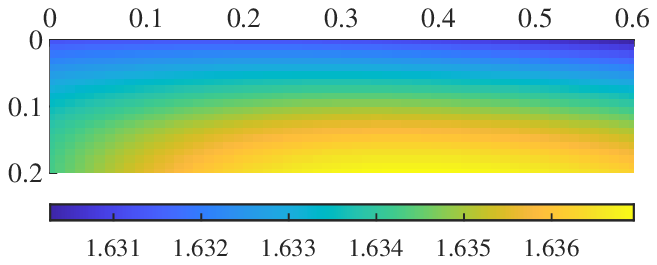}
         \caption{High-fidelity model}
         \label{f:Ex-2:ss_E_high}
     \end{subfigure}
     \begin{subfigure}[b]{0.495\textwidth}
         \centering
         \includegraphics[width=3in]{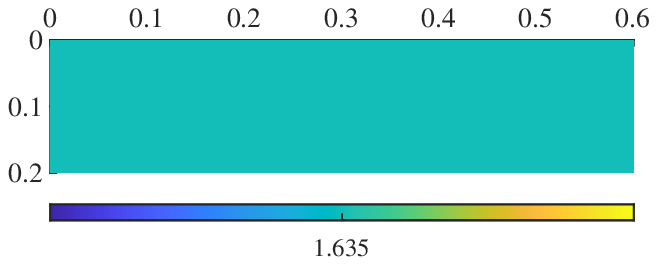}
         \caption{Low-fidelity model}
         \label{f:Ex-2:ss_E_low}
     \end{subfigure}
     \caption{Young's modulus for sample $\{\xi_1,\ldots,\xi_{n_{\text{KL}}}\}=\{0.5,\ldots,0.5\}$.}
     \label{f:Ex-2:ss_E}
\end{figure}
\begin{figure}[htb]
     \centering
     \begin{subfigure}[b]{0.495\textwidth}
         \centering
         \includegraphics[width=3in]{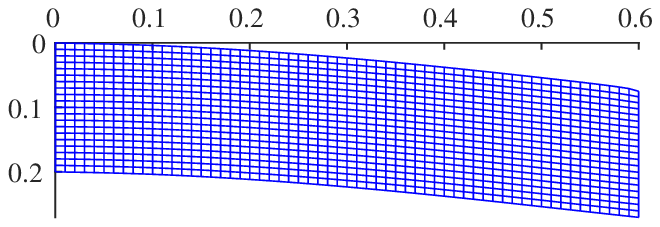}
         \caption{High-fidelity model}
         \label{f:Ex-2:ss_high}
     \end{subfigure}
     \begin{subfigure}[b]{0.495\textwidth}
         \centering
         \includegraphics[width=3in]{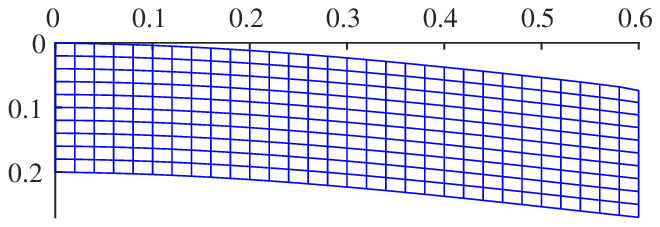}
         \caption{Low-fidelity model}
         \label{f:Ex-2:ss_low}
     \end{subfigure}
     \caption{Mesh deformation for sample $\{\xi_1,\ldots,\xi_{n_{\text{KL}}}\}=\{0.5,\ldots,0.5\}$.}
     \label{f:Ex-2:ss}
\end{figure}
\begin{figure}[ht]t
    \centering
    \includegraphics[width=2.5in]{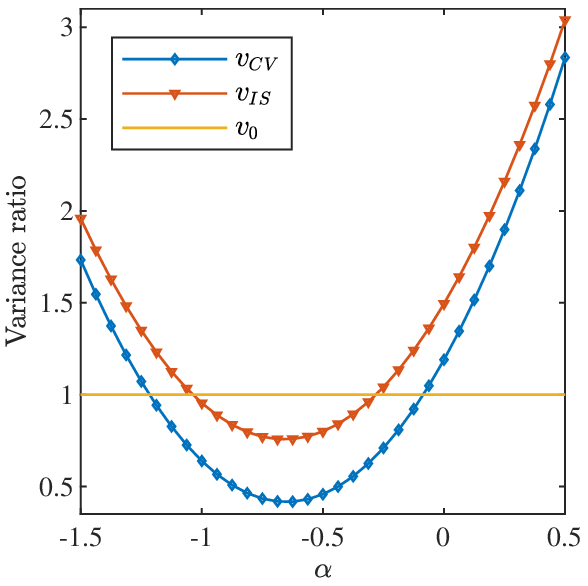}
    \caption{$\alpha$ vs. variance ratios for the cantilever beam.}
    \label{f:Ex-2:ve_v0}
\end{figure}
The equations are solved using the finite element method, whose mesh is assembled from square, linear, plane stress elements with unit thickness. Those elements are made of an isotropic, linear elastic material characterized by $E$ and $\nu$. 

We treat Young's modulus as the uncertain quantity in this problem. Young's modulus must be positive and finitely bounded and we model it as uniform random field transformed from a two-dimensional Gaussian random field with the following covariance
\begin{equation}
    \begin{aligned}
     K(\bm{s},\bm{t}) = 
\exp\left(\frac{-(s_1-t_1)^2}{r_1^2}\right)\exp\left(\frac{-(s_2-t_2)^2}{r_2^2}\right) = K_1(s_1,t_1)K_2(s_2,t_2) \textrm{ for }
    \bm{s},\bm{t} \in \mathcal{X},
    \label{e:covariance}
    \end{aligned}
\end{equation}
where $r_1$ and $r_2$ are the correlation lengths in the two coordinate directions. The transformation \cite{grigoriu_simulation_1998} is performed by
\begin{equation}
    \begin{aligned}
    E(\bm{x},\omega) = F^{-1} \circ [\Phi\left(y(\bm{x},\omega)\right)] \text{ for } \bm{x} \in \mathcal{X} \text{ and } \omega \in \Omega,
    \end{aligned}
\end{equation}
where $F^{-1}$ is the inverse of a prescribed CDF, $\Phi\left(y(\bm{x},\omega)\right)$ the standard normal CDF, $y(\bm{x},\omega)$ a stationary zero-mean Gaussian random field, and $E(\bm{x},\omega)$ the Young's modulus. Choosing $F^{-1}$ as the inverse of an uniform CDF, we have
\begin{equation}
    \begin{aligned}
        E(\bm{x},\omega) = a + (b-a)\Phi(y(\bm{x},\omega)),
        \label{e:modulus}
    \end{aligned}
\end{equation}
where $a$ and $b$ are the lower and upper bound of the uniform distribution. In practice the random field $y(\bm{x},\omega)$ needs to be discretized by an appropriate method such as Expansion Optimal Linear Estimator \cite{li_optimal_1993} and polynomial chaos expansion \cite{ghanem2003stochastic,xiu2010numerical}. Among such methods, the Karhunen-Lo\`{e}ve (KL) expansion minimizes the mean squared error \cite{alexanderian_brief_2015,wang_2008} resulting in the smallest number of terms in the expansion to obtain a required accuracy \cite{sudret_2000}. The KL expansion of the random field $y(\bm{x},\omega)$ is given as
\begin{equation}
    y(\bm{x},\omega) = \sum_{i=1}^\infty\sqrt{\lambda_i}\xi_i(\omega)\psi_i(\bm{x}),
    \label{e:kl_1}
\end{equation}
where $\xi_i(\omega)$ are standard normal random variables. The eigenvalues $\lambda_i$ and the corresponding orthogonal eigenfunctions $\psi_i(\bm{x})$ are solutions of the following eigenvalue problem:
\begin{equation}
    \int_{\mathcal{D}}K(\bm{s},\bm{t})\psi_i(\bm{t})d\bm{t}=\lambda_i \psi_i(\bm{s}),
    \label{e:fredholm}
\end{equation}
where $K(\bm{s},\bm{t})$ is the covariance function of the random field. Here two practical issues have to be considered. First, the infinite KL expansion in \eqref{e:kl_1} are truncated to be computable. Second, the integral equation \eqref{e:fredholm} is not trivial to solve on high-dimensional domain. Thus, the separability of the covariance function \eqref{e:covariance} is exploited leading to separable eigenvalues and eigenfunctions \cite{wang_2008} as follows
\begin{equation}
    \begin{aligned}
        \lambda_i &= \lambda_{i_1} \lambda_{i_2}, \\
        \psi_i(\bm{x}) &= \psi_{i_1}(x_1) \psi_{i_2}(x_2),
    \end{aligned}
\end{equation}
where $\lambda_{i_j}$ and $\psi_{i_j}(x_j)$, $j=\{1,2\}$, are the eigenvalues and eigenfunctions of the integral equation \eqref{e:fredholm} using the covariance function $K_j(s_j,t_j)$ in \eqref{e:covariance}; and $\lambda_i$ are arranged in decreasing order.

In this example the following numerical values are used: Poisson's ratio $\nu=0.3$; correlation lengths $r_1=60$ and $r_2=20$; beam dimensions $3L \times L = 0.6 \times 0.2$; bounds of the Young's modulus $a=1$ and $b=2.$

 The limit state functions are given as $g_i(u) = l_i - u,i=\{0,1\}$, where $u$ is the vertical displacement at the load application point. The thresholds $l_i$ are chosen so that the failure probabilities $P_f(g_i(u) < 0)$, which are computed using $10^6$ reference samples of $\{\xi_i(\omega)\}_{i=1}^{n_{\text{KL}}}$, are small and different, i.e., $l_0=118.923,l_1=108.510,P_f(g_0(u) < 0)=0.001173$ and $P_f(g_1(u) < 0)=0.022428$.

 We utilize two different meshing schemes as multi-fidelity models, i.e., the high-fidelity model corresponds to the mesh of $60 \times 20$ elements, and the low-fidelity model the mesh of $30 \times 10$ elements. The models also have different, but shared, sources of ucnertainty. For the high-fidelity model: $u$ depends on the Young's modulus $E(\bm{x},\omega)$ which in turn is transformed according to \eqref{e:modulus} using a truncated KL expansion $y(\bm{x},\omega)$ from \eqref{e:kl_1}; and the stiffness matrix of each finite element is calculated using the values of the eigenfunctions at the element's centroid coordinates. For the low-fidelity model, the Young's modulus is just a uniform random variable $E(\omega) = a + (b-a)\Phi(y(\omega))$, where
\begin{equation}
    \begin{aligned}
        y(\omega) = \sum_{i=1}^{n_{\text{KL}}}\sqrt{\lambda_i}\xi_i(\omega)\psi_i(\bar{\bm{x}})
    \end{aligned}
\end{equation}
and $\bar{\bm{x}}$ is the coordinates of the beam centroid, i.e., $\bar{\bm{x}}=(0.3,0.1)$. For both models, the number of $\lambda_{i_j}$ and $\psi_{i_j}(x_j)$ is $n_j=5$, and the number of $\lambda_i$ and $\psi_i(\bm{x})$ is $n_{\text{KL}}=10$. 

Figures \ref{f:Ex-2:ss_E} and \ref{f:Ex-2:ss} show the Young's modulus and the mesh deformation of the beam for the high- and low-fidelity model when all the random variables $\xi_i(\omega)$ take the value of $0.5$; the actual displacements are scaled down by $10^3$ for visualization. It is clear from Figure \ref{f:Ex-2:ss_E} that the Young's modulus is spatially variable for the high-fidelity model while it is only a constant for the low-fidelity model.

The EM algorithm uses inputs $n_s=5000, \tau=0.1$, and $k_{\text{init}}=5$. The baseline estimator $\mathcal{Q}^{\text{MFIS}}_{\hat{q},n}$ takes $n_0=4 \times 10^5$ samples to compute its variance. We first find that $\mathcal{C}_{\text{HF}} \approx 11 \mathcal{C}_{\text{LF}}$, where $\mathcal{C}_{\text{HF}}$ and $\mathcal{C}_{\text{LF}}$ are the costs to produce one evaluation of the LF and HF model, respectively. Then, we allocate the samples to each model as shown in Table~\ref{tbl:ex2_sample}, using the cost ratio such that the total cost of model evaluations is equal to that of the baseline estimator.
\begin{table}[h]
    \centering
    \begin{threeparttable}
    \begin{tabular}{ c c c c }
    \toprule
        & $v_0$ & $v_{\text{CV}}$ & $v_{\text{IS}}$ \\
    \cmidrule{1-4}
        $n_{\text{HF}}$ & 400000 & 366666 & 293333 \\
    \cmidrule{1-4}
        $n_{\text{LF}}$ & & 366666 & 1173332 \\
    \bottomrule
    \end{tabular}
    \caption{Sample allocation for the second example.}
    \label{tbl:ex2_sample}
    \end{threeparttable}
\end{table}

The results of this example are presented in Figure~\ref{f:Ex-2:ve_v0}. Again, our estimators perform better than the MFIS one and the CV estimator shows favorable result compared to the ACV scheme.

\subsection{Clamped Mindlin plate in bending}
\label{ss:ex-3}
The last example is a modified version of that provided in \cite{peherstorfer_multifidelity_2016}, where authors derive the MFIS estimator. While the MFIS estimator is able to achieve impressive speedups of up to several orders of magnitude compared to the MC method, we demonstrate that further variance reduction is still possible by employing control variates.

\begin{figure}[!htb]
    \centering
    \includegraphics[scale=0.8]{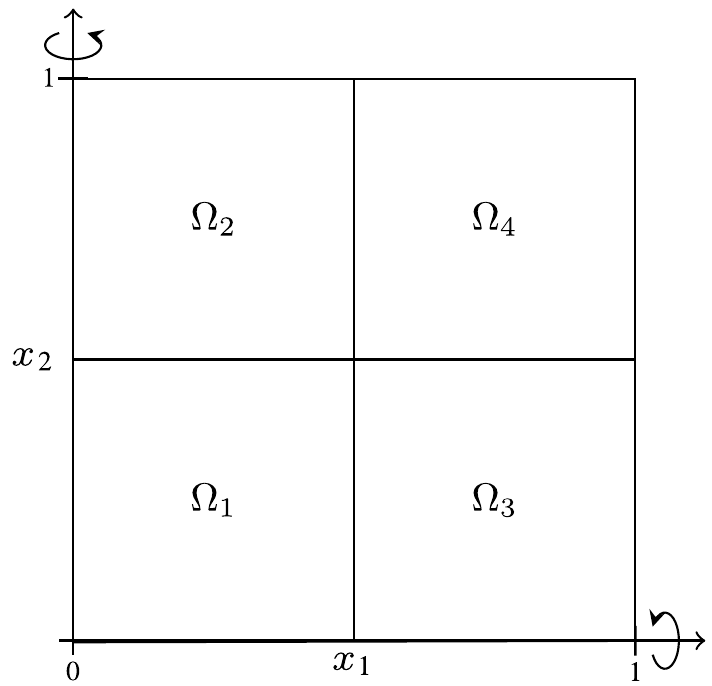}
    \caption{The Mindlin plate \cite{peherstorfer_multifidelity_2016}.}
    \label{fig:P}
\end{figure}
\begin{figure}[!htb]
     \centering
     \begin{subfigure}[b]{0.495\textwidth}
         \centering
         \includegraphics[width=2.5in]{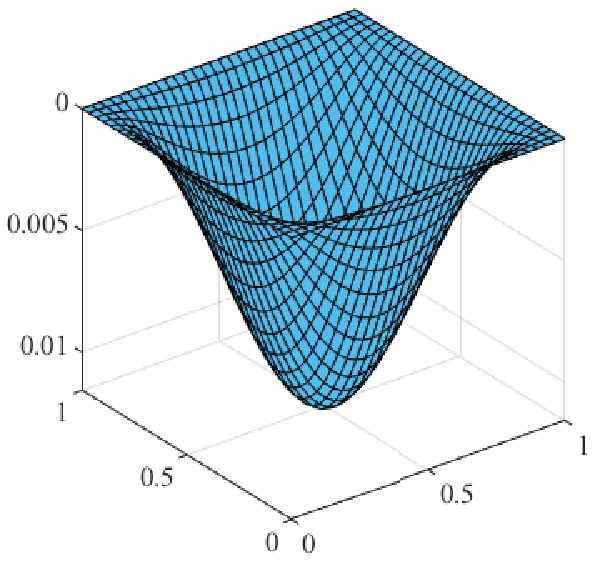}
         \caption{High-fidelity model}
         \label{f:Ex-3:ss_high}
     \end{subfigure}
     \begin{subfigure}[b]{0.495\textwidth}
         \centering
         \includegraphics[width=2.5in]{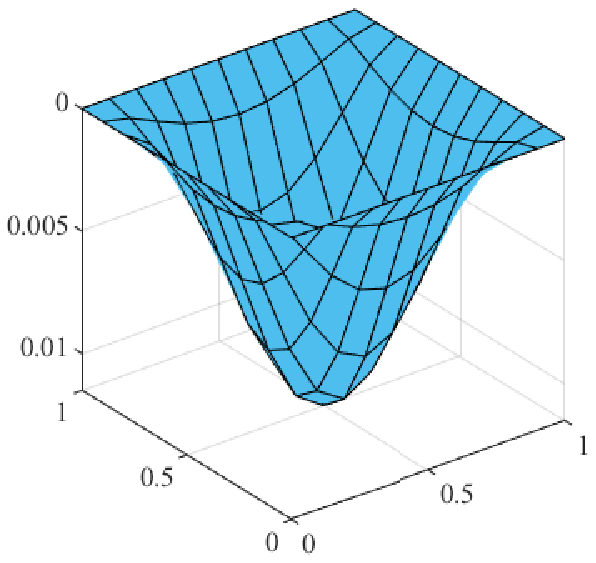}
         \caption{Low-fidelity model}
         \label{f:Ex-3:ss_low}
     \end{subfigure}
     \caption{Plate deformation for sample $\{h_i,s_i\}_{i=1}^4=\{0.05,1\}$.}
     \label{f:Ex-3:ss}
\end{figure}
Let $\mathcal{Y}=[0,1]^2$ denote the domain of the clamped Mindlin plate in Figure \ref{fig:P}; $\mathcal{E}$ the four edges of the plate; $\theta_1$ and $\theta_2$ the rotations of the normal to the plate middle plane with respect to the axes $x_2$ and $x_1$, respectively; and $w$ the displacement of the middle plane in the (out-of-plane) $x_3$-direction. The governing equations of the Mindlin's theory of plate in static equilibrium are given as
\begin{align}
    M_{ij,j}-Q_i &= 0, \nonumber \\
    Q_{i,i}+s &= 0, \nonumber \\
    \theta_1(\mathcal{E})=0,\theta_2(\mathcal{E})&=0,w(\mathcal{E})=0, \label{e:plate_bc}
\end{align}
where $M_{ij}$ and $Q_i$ are the moment and shear resultants, i.e., $M_{11}=D\left(\dfrac{\partial \theta_1}{\partial x_1} + \nu\dfrac{\partial \theta_2}{\partial x_2}\right)$, $Q_1=\kappa G h \left( \theta_1 + \dfrac{\partial w}{\partial x_1} \right)$, etc.; $D$, $G$, $\kappa$, $h$ and $\nu$ are the bending rigidity, shear modulus, shear correction factor, plate thickness and Poisson's ratio, respectively; and $s$ is a transverse load. We refer to \cite{lim_canonical_2003,reddy2006theory} for a complete treatment of the plate theory with detailed equations. The boundary conditions \eqref{e:plate_bc} state that for a clamped plate there are no rotations and displacement along the edges of the plate which is made of an isotropic, linear elastic material with Young's modulus $E=10^4$ and Poisson's ratio $\nu=0.3$. Following the settings in \cite{peherstorfer_multifidelity_2016}, we divide the plate into four regions $\{\mathcal{Y}_i\}_{i=1}^4$, each of which has a random thickness $h_i$ and is subject to a random load $s_i$. Both $h_i$ and $s_i$ are uniformly distributed, i.e., $\{h_i\}_{i=1}^4 \sim \mathcal{U}(0.05,0.1)$ and $\{s_i\}_{i=1}^4 \sim \mathcal{U}(1,2)$. According to the first-order shear deformation theory \cite{lim_canonical_2003} the shear correction factor is $\kappa=5/6$.
\begin{figure}[h]
    \centering
    \includegraphics[width=2.5in]{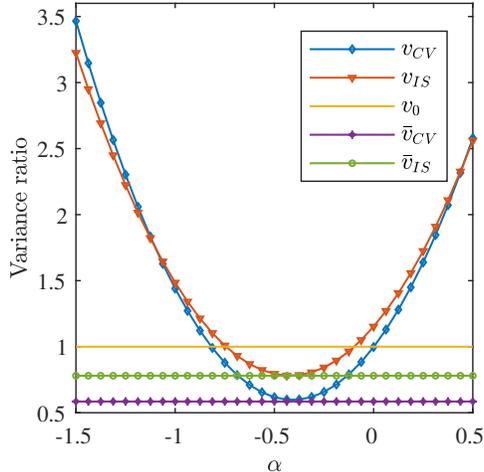}
    \caption{$\alpha$ vs. variance ratios for the Mindlin plate.}
    \label{f:Ex-3:Plate}
\end{figure}

In this example the high-fidelity model has the mesh size of $30 \times 30$ square bilinear isoparametric elements (i.e., $Q4$ elements) while the low-fidelity model utilizes $10 \times 10$ elements. The Matlab code from \cite[Chapter~12]{ferreira2008matlab} is adopted for finite element analysis. Figure \ref{f:Ex-3:ss} shows the plate deformation for both models using $\{h_i,s_i\}_{i=1}^4=\{0.05,1\}$. The limit state functions are defined as $g_i(\bm{z}) = l_i - w_i^{\text{c}}(\bm{z}),i=\{0,1\},$ where $\bm{z}=\left\{\{h_i\}_{i=1}^4,\{s_i\}_{i=1}^4\right\} \subset \mathbb{R}^8$ is a realization of input random variables and $w_i^{\text{c}}(\bm{z})$ is the $x_3$-direction displacement of the plate centroid. 

The parameters used in the EM algorithm are given as $n_s=5000, \tau=0.1$ and $k_{\text{init}}=5$. The baseline estimator $\mathcal{Q}^{\text{MFIS}}_{\hat{q},n}$ is evaluated using $n_0=4 \times 10^5$ samples for its variance. As shown in Table~\ref{tbl:ex3_sample}, the cost of other estimators is guaranteed to be equal to the baseline estimator by appropriate sample allocation using the number of ensembles $K=1000$ and the empirical formula $\mathcal{C}_{\text{HF}} \approx 37 \mathcal{C}_{\text{LF}}$.
\begin{table}[h]
    \centering
    \begin{threeparttable}
    \begin{tabular}{ c c c c }
    \toprule
        & $v_0$ & $v_{\text{CV}}$ & $v_{\text{IS}}$ \\
    \cmidrule{1-4}
        $n_{\text{HF}}$ & 400000 & 389473 & 356626 \\
    \cmidrule{1-4}
        $n_{\text{LF}}$ & & 389473 & 1604817 \\
    \bottomrule
    \end{tabular}
    \caption{Sample allocation for the third example.}
    \label{tbl:ex3_sample}
    \end{threeparttable}
\end{table}

The variance ratios are shown in Figure~\ref{f:Ex-3:Plate}. First, we stress that since the MFIS estimator does not take into account the correlations among models, it is not able to exploit the multi-fidelity modeling to the fullest extent, and, thus, our estimators have achieved clear advantages over it. Second, comparing $v_{\text{IS}}$ with $v_{\text{CV}}$, the estimator becomes less efficient with estimated weight as concluded in~\cite{gorodetsky_generalized_2020}. Third, as explained in the first example, $\bar{v}_{\text{CV}}$ and $\bar{v}_{\text{IS}}$ tightly follows the minimum of $v_{\text{CV}}$ and $v_{\text{IS}}$, respectively. Finally, it is noted that even with sub-optimal weight the performance of these estimators is still better than the MFIS estimator over a significant range of weights, as specified by Theorem~\ref{thm:range_ACV}.
\section{Conclusions}
\label{s:conclusions}

In this paper we have developed an ensemble estimator for approximate control variate schemes that provides a mechanism to estimate unknown covariances, in addition to unknown means.
 This contribution has allowed us to provide theoretical bounds on the number of samples required to guarantee certain variance reduction. Furthermore, this guarantee depends upon a correlation coefficient that is problem dependent. The second contribution is applying the framework in the context of importance sampling. We show that the approximate control variate can further reduce the variance compared to the MFIS approach described in~\cite{peherstorfer_multifidelity_2016}. We are able to achieve considerably greater variance reduction with this approach on several problems of computational mechanics.

Future work will seek to study values of the correlation coefficient that can be derived from the underlying problem---similar to what is done in multi-level MC for multi-fidelity models arising in varying discretizations. Another line of work is extending the importance sampling techniques to include several low-fidelity models. One challenge to overcome is effectively balancing the cost of computing a biasing distribution using the low-fidelity model and the variance reduction that it provides. Indeed, the current approaches to multi-fidelity importance sampling tend to ignore this computational aspect. Finally, as an effective variance reduction technique, our estimators have extensive application potential in expensive UQ problems including optimization under uncertainty, and in particular, reliability-based and robust optimization.

\section{Acknowledgements}
We thank Gianluca Geraci, John Jakeman, Mike Eldred, and Teresa Portone for helpful discussions surrounding this paper. This project was funded by the Sandia National Laboratories LDRD program.

\appendix
\numberwithin{equation}{section}

\section{The expectation-maximization algorithm}
\label{app:EM}
The Kullback-Leibler (KL) divergence between the optimal density $q^*(\bm{z})$ and the approximate density $\hat{q}(\bm{z})$ is
\begin{align*}
    \mathcal{D}(q^*(\bm{z}),\hat{q}(\bm{z})) &= 
    \mathbb{E}_{q^*}\left[ \ln\left( \frac{q^*(\bm{z})}{\hat{q}(\bm{z})} \right)\right] =
    \int_{\mathbb{R}^d} q^*(\bm{z}) \ln (q^*(\bm{z})) d\bm{z} - \int_{\mathbb{R}^d} q^*(\bm{z}) \ln (\hat{q}(\bm{z})) d\bm{z}.
\end{align*}
The CE method aims to minimize the KL divergence to find the unknowns in the GMM \eqref{e:GMM}. Let us gather the unknown parameters into the vector  $\bm{v}=\{\pi_i,\bm{\mu}_i,\bm{\Sigma}_i; i=1,2,\ldots,k\}$. Then the optimization problem can be equivalently written as
\begin{align}
    \min_{\bm{v}} \mathcal{D}(q^*(\bm{z}),\hat{q}(\bm{z};\bm{v})) = \max_{\bm{v}} \int_{\mathbb{R}^d} q^*(\bm{z}) \ln (\hat{q}(\bm{z};\bm{v})) d\bm{z},
    \label{e:min_KL}
\end{align}
because the first term of the KL divergence is independent of $\hat{q}.$
where $\hat{q}(\bm{z};\bm{v})$ stresses the presence of parameters in the GMM \eqref{e:GMM}. Replacing $Y_0(\bm{z})$ with $\mathcal{I}_{\mathcal{G}}(\bm{z})$ in \eqref{e:IS_density}, and inserting this expression into \eqref{e:min_KL} we obtain\footnote{$\mu_0$ is not needed since it is a constant.}
\begin{align}
    \min_{\bm{v}} \mathcal{D}(q^*(\bm{z}),\hat{q}(\bm{z};\bm{v})) = \max_{\bm{v}} \int_{\mathbb{R}^d} \mathcal{I}_{\mathcal{G}}(\bm{z}) p(\bm{z}) \ln (\hat{q}(\bm{z};\bm{v})) d\bm{z}.
    \label{e:min_KL_2}
\end{align}
 Another sampling density $\hat{q}(\bm{z};\bm{w})$, which has the same form as $\hat{q}(\bm{z};\bm{v})$ but with a different parameter vector $\bm{w}$, is introduced to facilitate the optimization algorithm
\begin{align}
    \min_{\bm{v}} \mathcal{D}(q^*(\bm{z}),\hat{q}(\bm{z};\bm{v})) &= \max_{\bm{v}} \int_{\mathbb{R}^d} \mathcal{I}_{\mathcal{G}}(\bm{z}) \ln (\hat{q}(\bm{z};\bm{v})) \widehat{W}(\bm{z};\bm{w}) \hat{q}(\bm{z};\bm{w}) d\bm{z} = \max_{\bm{v}} \mathbb{E}_{\hat{q}(\bm{z};\bm{w})} \left[\mathcal{I}_{\mathcal{G}}(\bm{z}) \ln (\hat{q}(\bm{z};\bm{v})) \widehat{W}(\bm{z};\bm{w})\right] \nonumber \\
    &\approx \max_{\bm{v}} \frac{1}{n_s} \sum_{i=1}^{n_s} \mathcal{I}_{\mathcal{G}}(\bm{z}_i) \ln (\hat{q}(\bm{z}_i;\bm{v})) \widehat{W}(\bm{z}_i;\bm{w}),
    \label{e:ce_with_samples}
\end{align}
where $\bm{z}_i,i=1,2,\ldots,n_s,$ are samples drawn from $\hat{q}(\bm{z};\bm{w}),$ and $\widehat{W}(\bm{z};\bm{w}) = \dfrac{p(\bm{z})}{\hat{q}(\bm{z};\bm{w})}$. It is noted that by choosing an appropriate joint likelihood $h(\hat{\bm{z}}|\bm{v}) = \prod_{i \in \hat{n}_s} \hat{q}(\bm{z}_i|\bm{v})^{\widehat{W}(\bm{z}_i)} $, the optimization problem \eqref{e:ce_with_samples} is equivalent to the maximum log-likelihood estimation (MLE) problem
\begin{align}
    \hat{\bm{v}} = \argmax_{\bm{v}} \ln (h(\hat{\bm{z}}|\bm{v})) = \argmax_{\bm{v}} \sum_{i \in \hat{n}_s} \ln (\hat{q}(\bm{z}_i|\bm{v})) \widehat{W}(\bm{z}_i),
\end{align}
where $\hat{\bm{z}}=\{\bm{z}_i\}_{i \in \hat{n}_s}$, $\hat{n}_s = \{i \in n_s: \mathcal{I}_{\mathcal{G}}(\bm{z}_i) \neq 0 \}$. The EM algorithm is an iterative method to find $\hat{\bm{v}}$, which is also the solution of \eqref{e:ce_with_samples}. Let $\hat{\bm{v}}^{(m)}$ denote the parameter vector at $m^{\text{th}}$ iteration. In \cite{Borman_EM} it is shown that 
\begin{align}
    \hat{\bm{v}}^{(m+1)} = \argmax_{\bm{v}} \mathbb{E}_{\bm{X}|\hat{\bm{z}},\hat{\bm{v}}^{(m)}} \left[ \ln h(\bm{x}|\bm{v}) \right] = \argmax_{\bm{v}} Q(\bm{v}|\hat{\bm{v}}^{(m)}),
    \label{e:update_EM}
\end{align}
where $\bm{X}$ is the complete data set. Using the GMM \eqref{e:GMM} and $n_{\tau}=\tau \hat{n}_s$, where $\tau \in ]0,1[$ is a fixed value to identify the intermediate failure domains, we have \cite{geyer_cross_2019,Chen_emdemystified} 
\begin{align}
    Q(\bm{v}|\hat{\bm{v}}^{(m)}) &= \sum_{i=1}^{n_{\tau}} \widehat{W}(\bm{z}_i) \sum_{j=1}^k \gamma_{ij}^{(m)}\ln (\pi_j \mathcal{N}(\bm{z}_i;\bm{\mu}_j,\bm{\Sigma}_j)) \label{e:Q_vv} \\
    \gamma_{ij}^{(m)} &= \frac{\pi_j^{(m)} \mathcal{N}(\bm{z}_i;\bm{\mu}_j^{(m)},\bm{\Sigma}_j^{(m)})}{\sum_{r=1}^k \pi_r^{(m)} \mathcal{N}(\bm{z}_i;\bm{\mu}_r^{(m)},\bm{\Sigma}_r^{(m)})} \label{e:gamma}.s
\end{align}
The updating scheme is then derived by solving the following optimization problem
\begin{equation}
    \begin{aligned}
        \max_{\bm{v}}     && & Q(\bm{v}|\hat{\bm{v}}^{(m)}) \\
        \text{subject to} && & \sum_{i=1}^k \pi_i = 1, \\
                          && & \pi_i \geq 0, i=1,2,\ldots,k, \\
                          && & \bm{\Sigma}_i \succ 0, i=1,2,\ldots,k,
    \end{aligned}
\end{equation}
where the first and second constraint enforce $\pi_i$ to be probabilities, and the last constraint is meant to render the covariance matrices positive definite. Using \eqref{e:Q_vv}, \eqref{e:gamma}, and the method of Lagrange multipliers, the updating equations \cite{geyer_cross_2019,Chen_emdemystified} are listed below.
\begin{align*}
    \nu_j^{(m)} &= \sum_{i=1}^{n_{\tau}} \widehat{W}(\bm{z}_i)\gamma_{ij}^{(m)}, \\
    \pi_j^{(m+1)} &= \frac{\nu_j^{(m)}}{\sum_{r=1}^{k} \nu_r^{(m)}}, \\
    \bm{\mu}_j^{(m+1)} &= \frac{\sum_{i=1}^{n_{\tau}} \widehat{W}(\bm{z}_i)\gamma_{ij}^{(m)} \bm{z}_i}{\nu_j^{(m)}}, \\
    \bm{\Sigma}_j^{(m+1)} &= \frac{\sum_{i=1}^{n_{\tau}} \widehat{W}(\bm{z}_i)\gamma_{ij}^{(m)}(\bm{z}_i - \bm{\mu}_j^{(m+1)})(\bm{z}_i - \bm{\mu}_j^{(m+1)})^{T}}{\nu_j^{(m)}}, j=1,2,\ldots,k.
\end{align*}

\section{Proof of Proposition~\ref{pro:alpha_sample_cov}}
\label{app:alpha_acv_sample}
We rewrite $\hat{\bm{\mathsf{C}}}$ and $\hat{\bm{\mathsf{c}}}$ using~\eqref{e:cv_sample_covariance} and~\eqref{e:centered_data_matrix} as
\begin{align}
    \hat{\bm{\mathsf{C}}} &= \frac{\bm{D}\bm{D}^{\text{T}}}{K-1}, \label{e:covar_mat_hat_D}  \\
    \hat{\bm{\mathsf{c}}} &= \frac{\bm{D}\left( \bm{Q}(Y_0) - \bar{\mathcal{Q}}_n(Y_0) \bm{1}_K \right)}{K-1}. \label{e:covar_vec_hat_D}
\end{align}
Substituting~\eqref{e:covar_mat_hat_D} and~\eqref{e:covar_vec_hat_D} into~\eqref{e:cv_sample_alpha},~\eqref{e:alpha_star_is}, and \eqref{e:alpha_star_mf}, we obtain
\begin{align}
    \ubar{\bm{\alpha}}_{\text{CV}} &= -\hat{\bm{\mathsf{C}}}^{-1}\hat{\bm{\mathsf{c}}} = -\left(\frac{\bm{D}\bm{D}^{\text{T}}}{K-1} \right)^{-1} \frac{\bm{D}\left( \bm{Q}(Y_0) - \bar{\mathcal{Q}}_n(Y_0)\bm{1}_K \right)}{K-1} = -\left(\bm{D}\bm{D}^{\text{T}}\right)^{-1} \bm{D}\left( \bm{Q}(Y_0) - \bar{\mathcal{Q}}_n(Y_0)\bm{1}_K \right), 
\end{align}
and
\begin{align}
    \ubar{\bm{\alpha}}_e &= -\left[\hat{\bm{\mathsf{C}}} \circ \bm{F}^e \right]^{-1}\left[\text{diag}\left(\bm{F}^e\right) \circ \hat{\bm{\mathsf{c}}} \right]
    = -\left[\frac{\bm{D}\bm{D}^{\text{T}}}{K-1} \circ \bm{F}^e \right]^{-1}\left[\text{diag}\left(\bm{F}^e\right) \circ \frac{\bm{D}\left( \bm{Q}(Y_0) - \bar{\mathcal{Q}}_n(Y_0) \bm{1}_K \right)}{K-1} \right] \nonumber \\
    &= -\left[\left(\bm{D}\bm{D}^{\text{T}}\right) \circ \bm{F}^e \right]^{-1}\left[\text{diag}\left(\bm{F}^e\right) \circ \left(\bm{D}\left( \bm{Q}(Y_0) - \bar{\mathcal{Q}}_n(Y_0) \bm{1}_K \right)\right) \right],
\end{align}
respectively.
Because $\bm{D}$ is the centered data matrix, we obtain
\begin{equation}
    \begin{aligned}
        \bm{D}\bm{1}_K &= \begin{bmatrix}
            \mathcal{Q}_n^{(1)}(Y_1) - \bar{\mathcal{Q}}_n(Y_1) & \mathcal{Q}_n^{(2)}(Y_1) - \bar{\mathcal{Q}}_n(Y_1) & \ldots & \mathcal{Q}_n^{(K)}(Y_1) - \bar{\mathcal{Q}}_n(Y_1) \\
            \mathcal{Q}_n^{(1)}(Y_2) - \bar{\mathcal{Q}}_n(Y_2) & \mathcal{Q}_n^{(2)}(Y_2) - \bar{\mathcal{Q}}_n(Y_2) & \ldots & \mathcal{Q}_n^{(K)}(Y_2) - \bar{\mathcal{Q}}_n(Y_2) \\
            \ldots & \ldots & \ldots & \ldots \\
            \mathcal{Q}_n^{(1)}(Y_M) - \bar{\mathcal{Q}}_n(Y_M) & \mathcal{Q}_n^{(2)}(Y_M) - \bar{\mathcal{Q}}_n(Y_M) & \ldots & \mathcal{Q}_n^{(K)}(Y_M) - \bar{\mathcal{Q}}_n(Y_M)
        \end{bmatrix}
        \begin{bmatrix}
            1 \\
            1 \\
            \ldots \\
            1
        \end{bmatrix} = \bm{0}_M.
        \label{e:D_1_K}
    \end{aligned}
\end{equation}
Since $\bm{D}\bm{1}_K=\bm{0}_M$, $\ubar{\bm{\alpha}}_{\text{CV}}$ and $\ubar{\bm{\alpha}}_e$ become
\begin{align*}
    \ubar{\bm{\alpha}}_{\text{CV}} &= -\left( \bm{D}\bm{D}^{\text{T}} \right)^{-1} \bm{D} \bm{Q}(Y_0), \\
    \ubar{\bm{\alpha}}_e &= -\left[\left(\bm{D}\bm{D}^{\text{T}}\right) \circ \bm{F}^e \right]^{-1}\left[\text{diag}\left(\bm{F}^e\right) \circ \left(\bm{D} \bm{Q}(Y_0)\right) \right].
\end{align*}

\section{Useful matrix algebra identities}
\label{app:identities}
The below proposition provides several identities to manipulate the expressions of the variances of the ensemble estimators.
\begin{proposition}
    \label{pro:identities}
    Let $\bm{A} \in \mathbb{R}^{M \times M}$, $\bm{B} \in \mathbb{R}^{M \times K}$, $\bm{V} \in \mathbb{R}^{K \times M}$ and $\bm{v} \in \mathbb{R}^K$. Then, we have the following identities
    \begin{align}
        \text{diag}(\bm{A}) \circ (\bm{B}\bm{v}) &= \left(\left(\text{diag}(\bm{A}) \otimes \bm{1}_K\right) \circ \bm{B} \right)\bm{v}, \label{e:diag_v} \\
        \left(\left(\text{diag}(\bm{A}) \otimes \bm{1}_K\right) \circ \bm{B}\right)\bm{V} &= \left(\text{diag}(\bm{A}) \otimes \bm{1}_M\right) \circ \left(\bm{B}\bm{V}\right), \\
        \bm{V}^{\text{T}}\left(\bm{B} \circ \left(\text{diag}(\bm{A}) \otimes \bm{1}_K\right) \right)^{\text{T}} &= \left(\bm{V}^{\text{T}}\bm{B}^{\text{T}}\right) \circ \left(\text{diag}(\bm{A}) \otimes \bm{1}_M\right)^{\text{T}}, \\
        \left(\left(\text{diag}(\bm{A}) \otimes \bm{1}_K\right) \circ \bm{A}\right) \left(\left(\text{diag}(\bm{A}) \otimes \bm{1}_K\right) \circ \bm{A}\right)^{\text{T}} &= \left(\bm{A}\bm{A}^{\text{T}}\right) \circ \left(\text{diag}(\bm{A}) \otimes \bm{1}_M\right) \circ \left(\text{diag}(\bm{A}) \otimes \bm{1}_M\right)^{\text{T}}, \label{e:diag_diag_T}
    \end{align}
    where $\circ$ is the Hadamard product and $\otimes$ is the outer product.
\end{proposition}
\begin{proof}
Let $[\star]_{i(ij)}$ denote an entry of the vector (matrix) $\star$. We prove the first two identities by showing the entries of both sides are equal. Thus,
\begin{align*}
    \left[\text{diag}(\bm{A}) \circ (\bm{B}\bm{v})\right]_i &= a_{ii}\sum_{j=1}^K b_{ij}v_j = \sum_{j=1}^K (a_{ii} b_{ij}) v_j = \left[\left(\left(\text{diag}(\bm{A}) \otimes \bm{1}_K\right) \circ \bm{B} \right)\bm{v}\right]_i \\
    \left[\left(\left(\text{diag}(\bm{A}) \otimes \bm{1}_K\right) \circ \bm{B}\right)\bm{V}\right]_{ij} &= \sum_{k=1}^K a_{ii}b_{ik} v_{kj} = a_{ii} \sum_{k=1}^K b_{ik} v_{kj} = \left[\left(\text{diag}(\bm{A}) \otimes \bm{1}_M\right) \circ \left(\bm{B}\bm{V}\right)\right]_{ij}
\end{align*}
The third identity is proved by transposing both sides of the second one as
\begin{align*}
    \left[\left(\left(\text{diag}(\bm{A}) \otimes \bm{1}_K\right) \circ \bm{B}\right)\bm{V}\right]^{\text{T}} &= \left[\left(\text{diag}(\bm{A}) \otimes \bm{1}_M\right) \circ \left(\bm{B}\bm{V}\right)\right]^{\text{T}} \\
     \bm{V}^{\text{T}}\left(\bm{B} \circ \left(\text{diag}(\bm{A}) \otimes \bm{1}_K\right) \right)^{\text{T}} &= \left(\bm{B}\bm{V}\right)^{\text{T}} \circ \left(\text{diag}(\bm{A}) \otimes \bm{1}_M\right)^{\text{T}} \\
     \bm{V}^{\text{T}}\left(\bm{B} \circ \left(\text{diag}(\bm{A}) \otimes \bm{1}_K\right) \right)^{\text{T}} &= \left(\bm{V}^{\text{T}}\bm{B}^{\text{T}}\right) \circ \left(\text{diag}(\bm{A}) \otimes \bm{1}_M\right)^{\text{T}}
\end{align*}
We find the final identity by applying the second and third one consecutively, and note that the Hadamard product is commutative and associative
\begin{align*}
    \left(\left(\text{diag}(\bm{A}) \otimes \bm{1}_K\right) \circ \bm{A}\right) \left(\left(\text{diag}(\bm{A}) \otimes \bm{1}_K\right) \circ \bm{A}\right)^{\text{T}} &= \left(\text{diag}(\bm{A}) \otimes \bm{1}_M\right) \circ \left(\bm{A}\left(\left(\text{diag}(\bm{A}) \otimes \bm{1}_K\right) \circ \bm{A}\right)^{\text{T}}\right) \\
    &= \left(\text{diag}(\bm{A}) \otimes \bm{1}_M\right) \circ \left(\left(\bm{A}\bm{A}^{\text{T}}\right) \circ \left(\text{diag}(\bm{A}) \otimes \bm{1}_M\right)^{\text{T}}\right) \\
    &= \left(\bm{A}\bm{A}^{\text{T}}\right) \circ \left(\text{diag}(\bm{A}) \otimes \bm{1}_M\right) \circ \left(\text{diag}(\bm{A}) \otimes \bm{1}_M\right)^{\text{T}}.
\end{align*}
\end{proof}

\section{Proof of Proposition~\ref{pro:q_mu_dist}}
\label{app:q_mu_dist}
The assumptions imply that the distributions of $\bar{\bm{\mathcal{Q}}} - \bar{\bm{\mu}}^{\text{CV}}$ and $\bar{\bm{\mathcal{Q}}} - \bar{\bm{\mu}}^{e}$ are multivariate normal. Thus, the proof focuses on finding the means and variances of those distributions. 

It is trivial to show that
\begin{equation*}
    \ee{}{\bar{\bm{\mathcal{Q}}} - \bar{\bm{\mu}}^{\text{CV}}} = \ee{}{\bar{\bm{\mathcal{Q}}} - \bar{\bm{\mu}}^{e}} = \bm{0}_M.
\end{equation*}
\renewcommand{\theenumi}{\alph{enumi}}
\begin{enumerate}
\item We compute the variance of $\bar{\bm{\mathcal{Q}}} - \bar{\bm{\mu}}^{\text{CV}}$ as
\begin{equation*}
    \begin{aligned}
        &\vvar{}{\bar{\bm{\mathcal{Q}}} - \bar{\bm{\mu}}^{\text{CV}}} = \vvar{}{\bar{\bm{\mathcal{Q}}}} \\
        &=
        \begin{bmatrix}
            \cov{}{\dfrac{1}{K} \displaystyle\sum_{j=1}^K \mathcal{Q}_n^{(j)}(Y_1),\dfrac{1}{K} \displaystyle\sum_{j=1}^K \mathcal{Q}_n^{(j)}(Y_1)} & \ldots & \cov{}{\dfrac{1}{K} \displaystyle\sum_{j=1}^K \mathcal{Q}_n^{(j)}(Y_1),\dfrac{1}{K} \displaystyle\sum_{j=1}^K \mathcal{Q}_n^{(j)}(Y_M)} \\
            \ldots & \ldots & \ldots \\
            \cov{}{\dfrac{1}{K} \displaystyle\sum_{j=1}^K \mathcal{Q}_n^{(j)}(Y_M),\dfrac{1}{K} \displaystyle\sum_{j=1}^K \mathcal{Q}_n^{(j)}(Y_1)} & \ldots & \cov{}{\dfrac{1}{K} \displaystyle\sum_{j=1}^K \mathcal{Q}_n^{(j)}(Y_M),\dfrac{1}{K} \displaystyle\sum_{j=1}^K \mathcal{Q}_n^{(j)}(Y_M)}
        \end{bmatrix} \\
        &= \frac{1}{K} \begin{bmatrix}
            \cov{}{\mathcal{Q}_n(Y_1),\mathcal{Q}_n(Y_1)} & \cov{}{\mathcal{Q}_n(Y_1),\mathcal{Q}_n(Y_2)} & \ldots & \cov{}{\mathcal{Q}_n(Y_1),\mathcal{Q}_n(Y_M)} \\
            \ldots & \ldots & \ldots & \ldots \\
            \cov{}{\mathcal{Q}_n(Y_M),\mathcal{Q}_n(Y_1)} & \cov{}{\mathcal{Q}_n(Y_M),\mathcal{Q}_n(Y_2)} & \ldots & \cov{}{\mathcal{Q}_n(Y_M),\mathcal{Q}_n(Y_M)}
        \end{bmatrix}
        = \frac{\bm{\mathsf{C}}}{K}.
    \end{aligned}
\end{equation*}
\item From \cite[Appendix D]{gorodetsky_generalized_2020} and \cite[Appendix E]{gorodetsky_generalized_2020}, we know that
\begin{align*}
    \vvar{}{\bar{\bm{\mathcal{Q}}} - \bar{\bm{\mu}}^{e}} &= \dfrac{\bm{\mathsf{C}} \circ \bm{F}^{e}}{K}.
\end{align*}
\end{enumerate}

\section{Proof of Theorem~\ref{thm:var_estimators}}
\label{app:var_estimators}
The goal of this proposition is to compute $\vvar{}{\bar{\mathcal{Q}}^{\text{CV}}(\ubar{\bm{\alpha}}_{\text{CV}})}$ and $\vvar{}{\bar{\mathcal{Q}}^e(\ubar{\bm{\alpha}}_e)}$ for $e \in \{\text{ACV-IS},\text{ACV-MF}\}$ in terms of $\vvar{}{\mathcal{Q}_n(Y_0)}$ and some known quantities, e.g., the covariances amongst models, the number of ensembles, etc.

We begin with an auxiliary result that will be used in the rest of the proof. Recall the law of total expectation 
\begin{align*}
    \ee{\mathsf{X}}{\mathsf{X}} = \ee{\mathsf{Y}}{\ee{\mathsf{X}}{\mathsf{X} \middle\vert \mathsf{Y}}},
\end{align*}
where $\mathsf{X}$ and $\mathsf{Y}$ are some random variables in the same probability space. We apply it to calculate the variances of $\bar{\mathcal{Q}}^{\text{CV}}(\ubar{\bm{\alpha}}_{\text{CV}})$ and $\bar{\mathcal{Q}}^e(\ubar{\bm{\alpha}}_e)$ by setting $\mathsf{X} = \left( \bar{\mathcal{Q}}^{\text{CV}}(\ubar{\bm{\alpha}}_{\text{CV}}) - \ee{}{\bar{\mathcal{Q}}^{\text{CV}}(\ubar{\bm{\alpha}}_{\text{CV}})} \right)^2$, $\mathsf{Y} = \bm{\mathcal{Q}}$ and $\mathsf{X} = \left( \bar{\mathcal{Q}}^e(\ubar{\bm{\alpha}}_e) - \ee{}{\bar{\mathcal{Q}}^e(\ubar{\bm{\alpha}}_e)} \right)^2$, $\mathsf{Y} = \tilde{\bm{\mathcal{Q}}}^e$ to obtain
\begin{align*}
    \vvar{}{\bar{\mathcal{Q}}^{\text{CV}}(\ubar{\bm{\alpha}}_{\text{CV}})} &= \ee{\bm{\mathcal{Q}}}{\vvar{}{\bar{\mathcal{Q}}^{\text{CV}}(\ubar{\bm{\alpha}}_{\text{CV}}) \middle\vert \bm{\mathcal{Q}}}}, \\
    \vvar{}{\bar{\mathcal{Q}}^e(\ubar{\bm{\alpha}}_e)} &= \ee{\tilde{\bm{\mathcal{Q}}}^e}{\vvar{}{\bar{\mathcal{Q}}^e(\ubar{\bm{\alpha}}_e) \middle\vert \bm{\mathcal{Q}}}},
\end{align*}
where $\bm{\mathcal{Q}} = \left[\mathcal{Q}_n(Y_1),\mathcal{Q}_n(Y_2),\ldots, \mathcal{Q}_n(Y_M)\right]^{\text{T}}$, $\tilde{\bm{\mathcal{Q}}}^e = \left[\mathcal{Q}_n(Y_1)-\mu_1^e,\mathcal{Q}_n(Y_2)-\mu_2^e,\ldots,\mathcal{Q}_n(Y_M)-\mu_M^e\right]^{\text{T}}$ and $\mu_i^e = \mathcal{Q}_{nr_i}(Y_i)$ for $i=1,2,\ldots,M$. As we can see, the vectors $\bm{\mathcal{Q}}$ and $\tilde{\bm{\mathcal{Q}}}^e$ only involve the low-fidelity models and the expectations with respect to these vectors eliminate the dependence of $\vvar{}{\bar{\mathcal{Q}}^{\text{CV}}(\ubar{\bm{\alpha}}_{\text{CV}})}$ and $\vvar{}{\bar{\mathcal{Q}}^e(\ubar{\bm{\alpha}}_e)}$ on $\{Y_i\}_{i=1}^M$. Since $\vvar{}{\bar{\mathcal{Q}}^{\text{CV}}(\ubar{\bm{\alpha}}_{\text{CV}}) \middle\vert \bm{\mathcal{Q}}}$ is a special case of $\vvar{}{\bar{\mathcal{Q}}^e(\ubar{\bm{\alpha}}_e) \middle\vert \tilde{\bm{\mathcal{Q}}}^e}$ with $\bm{F}^e = \bm{1}_M \otimes \bm{1}_M$, the computation of the later plays a central role in the proof below.

We now begin the main logic of the proof by substituting~\eqref{e:alpha_acv_sample} into~\eqref{e:acv_sample_mean} to obtain
\begin{equation}
    \begin{aligned}
        \bar{\mathcal{Q}}^e(\ubar{\bm{\alpha}}_e) &=  \frac{\bm{1}_K^{\text{T}}}{K} \bm{Q}(Y_0) - \left(\bar{\bm{\mathcal{Q}}} - \bar{\bm{\mu}}^e\right)^{\text{T}} \left[\left(\bm{D}\bm{D}^{\text{T}}\right) \circ \bm{F}^e \right]^{-1}\left[\text{diag}\left(\bm{F}^e\right) \circ \left(\bm{D} \bm{Q}(Y_0)\right) \right] \\
        &=  \left( \frac{\bm{1}_K^{\text{T}}}{K} - \left(\bar{\bm{\mathcal{Q}}} - \bar{\bm{\mu}}^e\right)^{\text{T}} \left[\left(\bm{D}\bm{D}^{\text{T}}\right) \circ \bm{F}^e \right]^{-1}\left[\bm{\mathcal{F}}^e_K \circ \bm{D} \right] \right) \bm{Q}(Y_0) \\
        &= \bm{\mathcal{X}}^{\text{T}} \bm{Q}(Y_0),
    \end{aligned}
    \label{e:x_x_s_y0_acv}
\end{equation}
where $\bm{\mathcal{F}}^e_K = \text{diag}(\bm{F}^e) \otimes \bm{1}_K$. The second line of~\eqref{e:x_x_s_y0_acv} uses the identity~\eqref{e:diag_v}.

Given $\tilde{\bm{\mathcal{Q}}}^e$, $\bm{\mathcal{X}}$ in~\eqref{e:x_x_s_y0_acv} is fixed so that the conditional variance of $\bar{\mathcal{Q}}^e(\ubar{\bm{\alpha}}_e)$ becomes
\begin{equation}
    \begin{aligned}
        \vvar{}{\bar{\mathcal{Q}}^e(\ubar{\bm{\alpha}}_e) \middle\vert \tilde{\bm{\mathcal{Q}}}^e} = \bm{\mathcal{X}}^{\text{T}} \vvar{}{\bm{Q}(Y_0) \middle\vert \tilde{\bm{\mathcal{Q}}}^e} \bm{\mathcal{X}}.
        \label{e:var_q_xTQxb}
    \end{aligned}
\end{equation}
To compute $\vvar{}{\bm{Q}(Y_0) \middle\vert \tilde{\bm{\mathcal{Q}}}^e}$, we utilize the assumption that the vector $\{\mathcal{Q}_n(Y_0),\mathcal{Q}_n(Y_1),\ldots,\mathcal{Q}_n(Y_M),\mu_1^e,\ldots,\mu_M^e\}$ has a multivariate normal distribution, and so the distribution of $\mathcal{Q}_n(Y_0)$ conditional on $\tilde{\bm{\mathcal{Q}}}^e = \{\mathcal{Q}_n(Y_1)-\mu_1^e,\ldots,\mathcal{Q}_n(Y_M)-\mu_M^e\}$ is also multivariate normal \cite[Theorem 5.3]{h2019applied} with variance
\begin{equation}
    \begin{aligned}
        \vvar{}{\mathcal{Q}_n(Y_0) \middle\vert \tilde{\bm{\mathcal{Q}}}^e} &= \left( 1- \cov{}{\tilde{\bm{\mathcal{Q}}}^e,\mathcal{Q}_n(Y_0)}^{\text{T}} \frac{\cov{}{\tilde{\bm{\mathcal{Q}}}^e,\tilde{\bm{\mathcal{Q}}}^e}^{-1}}{\vvar{}{\mathcal{Q}_n(Y_0)}} \cov{}{\tilde{\bm{\mathcal{Q}}}^e,\mathcal{Q}_n(Y_0)} \right) \vvar{}{\mathcal{Q}_n(Y_0)} \\
        &= (1-R^2_e) \vvar{}{\mathcal{Q}_n(Y_0)}.
    \end{aligned}
\end{equation}
Then, the conditional variance of $\bm{Q}(Y_0)$ given $\tilde{\bm{\mathcal{Q}}}^e$ is
\begin{equation}
    \begin{aligned}
        &\vvar{}{\bm{Q}(Y_0) \middle\vert \tilde{\bm{\mathcal{Q}}}^e} \\
        &=
        \begin{bmatrix}
            \vvar{}{\mathcal{Q}_n^{(1)}(Y_0) \middle\vert \tilde{\bm{\mathcal{Q}}}^e} & \cov{}{\left( \mathcal{Q}_n^{(1)}(Y_0),\mathcal{Q}_n^{(2)}(Y_0) \right) \middle\vert \tilde{\bm{\mathcal{Q}}}^e} & \ldots & \cov{}{\left( \mathcal{Q}_n^{(1)}(Y_0),\mathcal{Q}_n^{(K)}(Y_0) \right) \middle\vert \tilde{\bm{\mathcal{Q}}}^e} \\
            \cov{}{\left( \mathcal{Q}_n^{(2)}(Y_0),\mathcal{Q}_n^{(1)}(Y_0) \right) \middle\vert \tilde{\bm{\mathcal{Q}}}^e} & \vvar{}{\mathcal{Q}_n^{(2)}(Y_0) \middle\vert \tilde{\bm{\mathcal{Q}}}^e} & \ldots & \cov{}{\left( \mathcal{Q}_n^{(2)}(Y_0),\mathcal{Q}_n^{(K)}(Y_0) \right) \middle\vert \tilde{\bm{\mathcal{Q}}}^e} \\
            \ldots & \ldots & \ldots & \ldots \\
            \cov{}{\left( \mathcal{Q}_n^{(K)}(Y_0),\mathcal{Q}_n^{(1)}(Y_0) \right) \middle\vert \tilde{\bm{\mathcal{Q}}}^e} & \cov{}{\left( \mathcal{Q}_n^{(K)}(Y_0),\mathcal{Q}_n^{(2)}(Y_0) \right) \middle\vert \tilde{\bm{\mathcal{Q}}}^e}  & \ldots & \vvar{}{\mathcal{Q}_n^{(K)}(Y_0) \middle\vert \tilde{\bm{\mathcal{Q}}}^e}
        \end{bmatrix} \\
        &=
        \begin{bmatrix}
            \vvar{}{\mathcal{Q}_n^{(1)}(Y_0) \middle\vert \tilde{\bm{\mathcal{Q}}}^e} & 0 & \ldots & 0 \\
            0 & \vvar{}{\mathcal{Q}_n^{(2)}(Y_0) \middle\vert \tilde{\bm{\mathcal{Q}}}^e} & \ldots & 0 \\
            \ldots & \ldots & \ldots & \ldots \\
            0 & 0  & \ldots & \vvar{}{\mathcal{Q}_n^{(K)}(Y_0) \middle\vert \tilde{\bm{\mathcal{Q}}}^e}
        \end{bmatrix} \\
        &= \vvar{}{\mathcal{Q}_n(Y_0) \middle\vert \tilde{\bm{\mathcal{Q}}}^e} \bm{I} = (1-R^2_e) \vvar{}{\mathcal{Q}_n(Y_0)} \bm{I},
        \label{e:cond_var_SY0b}
    \end{aligned}
\end{equation}
where $\bm{I} \in \mathbb{R}^{K \times K}$ is the identity matrix. The second equality arises due to the i.i.d assumption of each of the $K$ batches.

Substituting~\eqref{e:cond_var_SY0b} into~\eqref{e:var_q_xTQxb}, we obtain
\begin{equation}
    \begin{aligned}
        \vvar{}{\bar{\mathcal{Q}}^e(\ubar{\bm{\alpha}}_e) \middle\vert \tilde{\bm{\mathcal{Q}}}^e} = \vvar{}{\mathcal{Q}_n(Y_0)}(1-R^2_e) \bm{\mathcal{X}}^{\text{T}}\bm{\mathcal{X}}
    \end{aligned}
\end{equation}
The expression $\bm{\mathcal{X}}^{\text{T}}\bm{\mathcal{X}}$ can be expanded as
\begin{equation}
    \begin{aligned}
        &\bm{\mathcal{X}}^{\text{T}}\bm{\mathcal{X}} =
        \left( \frac{\bm{1}_K^{\text{T}}}{K} - \left(\bar{\bm{\mathcal{Q}}} - \bar{\bm{\mu}}^e\right)^{\text{T}} \left[\left(\bm{D}\bm{D}^{\text{T}}\right) \circ \bm{F}^e \right]^{-1}\left[\bm{\mathcal{F}}^e_K \circ \bm{D} \right] \right) \left( \frac{\bm{1}_K^{\text{T}}}{K} - \left(\bar{\bm{\mathcal{Q}}} - \bar{\bm{\mu}}^e\right)^{\text{T}} \left[\left(\bm{D}\bm{D}^{\text{T}}\right) \circ \bm{F}^e \right]^{-1}\left[\bm{\mathcal{F}}^e_K \circ \bm{D} \right] \right)^{\text{T}} \\
        &= \left( \frac{\bm{1}_K^{\text{T}}}{K} - \left(\bar{\bm{\mathcal{Q}}} - \bar{\bm{\mu}}^e\right)^{\text{T}} \left[\left(\bm{D}\bm{D}^{\text{T}}\right) \circ \bm{F}^e \right]^{-1}\left[\bm{\mathcal{F}}^e_K \circ \bm{D} \right] \right) \left( \frac{\bm{1}_K}{K} - \left[\bm{\mathcal{F}}^e_K \circ \bm{D} \right]^{\text{T}} \left[\left(\bm{D}\bm{D}^{\text{T}}\right) \circ \bm{F}^e \right]^{-1} \left(\bar{\bm{\mathcal{Q}}} - \bar{\bm{\mu}}^e\right) \right) \\
        &= \frac{1}{K} - \frac{\bm{1}_K^{\text{T}}\left[\bm{\mathcal{F}}^e_K \circ \bm{D} \right]^{\text{T}}}{K}\mathcal{A} - \mathcal{A}^{\text{T}} \frac{\left[\bm{\mathcal{F}}^e_K \circ \bm{D} \right]\bm{1}_K}{K} + \mathcal{A}^{\text{T}} \left[\bm{\mathcal{F}}^e_K \circ \bm{D} \right] \left[\bm{\mathcal{F}}^e_K \circ \bm{D} \right]^{\text{T}} \mathcal{A},
        \label{e:var_is_mf_1}
    \end{aligned}
\end{equation}
where $\mathcal{A} = \left[\left(\bm{D}\bm{D}^{\text{T}}\right) \circ \bm{F}^e \right]^{-1} \left(\bar{\bm{\mathcal{Q}}} - \bar{\bm{\mu}}^e\right)$. Using the identity~\eqref{e:diag_v}, we obtain
\begin{align*}
    \left[\bm{\mathcal{F}}^e_K \circ \bm{D} \right]\bm{1}_K = \left[ \left(\text{diag}(\bm{F}^e) \otimes \bm{1}_K\right) \circ \bm{D} \right] \bm{1}_K = \text{diag}(\bm{F}^e) \circ (\bm{D}\bm{1}_K) &= \bm{0}_M, \\
    \bm{1}_K^{\text{T}}\left[\bm{\mathcal{F}}^e_K \circ \bm{D} \right]^{\text{T}} &= \bm{0}_M^{\text{T}},
\end{align*}
which simplify~\eqref{e:var_is_mf_1} as
\begin{equation}
    \begin{aligned}
        \bm{\mathcal{X}}^{\text{T}}\bm{\mathcal{X}} &= \frac{1}{K} + \mathcal{A}^{\text{T}} \left[\bm{\mathcal{F}}^e_K \circ \bm{D} \right] \left[\bm{\mathcal{F}}^e_K \circ \bm{D} \right]^{\text{T}} \mathcal{A} \\
        &= \frac{1}{K} + \left(\bar{\bm{\mathcal{Q}}} - \bar{\bm{\mu}}^e\right)^{\text{T}} \left[\left(\bm{D}\bm{D}^{\text{T}}\right) \circ \bm{F}^e \right]^{-1} \left[\bm{\mathcal{F}}^e_K \circ \bm{D} \right] \left[\bm{\mathcal{F}}^e_K \circ \bm{D} \right]^{\text{T}}
        \left[\left(\bm{D}\bm{D}^{\text{T}}\right) \circ \bm{F}^e \right]^{-1}
        \left(\bar{\bm{\mathcal{Q}}} - \bar{\bm{\mu}}^e\right) \\
        &= \frac{1}{K} + \left(\bar{\bm{\mathcal{Q}}} - \bar{\bm{\mu}}^e\right)^{\text{T}} \left[\left(\bm{D}\bm{D}^{\text{T}}\right) \circ \bm{F}^e \right]^{-1}
        \left[ \left(\bm{D}\bm{D}^{\text{T}}\right) \circ \bm{\mathcal{F}}^e_M \circ \left(\bm{\mathcal{F}}^e_M\right)^{\text{T}} \right]
        \left[\left(\bm{D}\bm{D}^{\text{T}}\right) \circ \bm{F}^e \right]^{-1}
        \left(\bar{\bm{\mathcal{Q}}} - \bar{\bm{\mu}}^e\right)
        \label{e:var_is_mf_2}
    \end{aligned}
\end{equation}
where the last equality of~\eqref{e:var_is_mf_2} applies the identity~\eqref{e:diag_diag_T}. Thus, Theorem~\ref{thm:var_estimators}b is proved as
\begin{equation}
    \begin{aligned}
        \vvar{}{\bar{\mathcal{Q}}^e(\ubar{\bm{\alpha}}_e)} = \;&\ee{\tilde{\bm{\mathcal{Q}}}^e}{\vvar{}{\bar{\mathcal{Q}}^e(\ubar{\bm{\alpha}}_e) \middle\vert \tilde{\bm{\mathcal{Q}}}^e}} = \vvar{}{\mathcal{Q}_n(Y_0)}(1-R^2_e) \\
        &\times \left( \frac{1}{K} + \ee{\tilde{\bm{\mathcal{Q}}}^e}{\left(\bar{\bm{\mathcal{Q}}} - \bar{\bm{\mu}}^e\right)^{\text{T}} \left[\left(\bm{D}\bm{D}^{\text{T}}\right) \circ \bm{F}^e \right]^{-1}
        \left[ \left(\bm{D}\bm{D}^{\text{T}}\right) \circ \bm{\mathcal{F}}^e_M \circ \left(\bm{\mathcal{F}}^e_M\right)^{\text{T}} \right]
        \left[\left(\bm{D}\bm{D}^{\text{T}}\right) \circ \bm{F}^e \right]^{-1}
        \left(\bar{\bm{\mathcal{Q}}} - \bar{\bm{\mu}}^e\right)} \right)
    \end{aligned}
\end{equation}
To deduce the result of Theorem~\ref{thm:var_estimators}a, we replace $\bar{\mathcal{Q}}^e(\ubar{\bm{\alpha}}_e)$, $R^2_e$, $\tilde{\bm{\mathcal{Q}}}^e$ and $\bar{\bm{\mu}}^e$ with $\bar{\mathcal{Q}}^{\text{CV}}(\ubar{\bm{\alpha}}_{\text{CV}})$, $R^2$, $\bm{\mathcal{Q}}$ and $\bar{\bm{\mu}}^{\text{CV}}$, respectively, and note that $\bm{F}^e = \bm{\mathcal{F}}^e_M = \bm{1}_M \otimes \bm{1}_M$ in the CV case; hence,
\begin{align}
    \vvar{}{\bar{\mathcal{Q}}^{\text{CV}}(\ubar{\bm{\alpha}}_{\text{CV}})}
    &= \vvar{}{\mathcal{Q}_n(Y_0)}\left(1-R^2\right) \left( \frac{1}{K} + \ee{\bm{\mathcal{Q}}}{\left(\bar{\bm{\mathcal{Q}}} - \bar{\bm{\mu}}^{\text{CV}}\right)^{\text{T}} \left( \bm{D}\bm{D}^{\text{T}}\right)^{-1} \left(\bar{\bm{\mathcal{Q}}} - \bar{\bm{\mu}}^{\text{CV}}\right)} \right).
\end{align}

\section{Proof of Theorem~\ref{thm:var_cv_sample_weight}}
\label{app:var_cv_sample_weight}

Proposition~\ref{pro:q_mu_dist} suggests that the expressions inside the expectation operators in Theorem~\ref{thm:var_estimators} may follow the Hotelling's $T^2$ distributions. It is indeed the case for~\eqref{e:var_q_cv}, while extra assumptions on the ACV-IS and ACV-MF schemes are needed to establish the distribution in~\eqref{e:var_q_acv}. The means of the Hotelling's $T^2$ distributions are then computed explicitly to prove Theorem~\ref{thm:var_cv_sample_weight}. We note that~\eqref{e:var_ratio_cv} is a special case of~\eqref{e:var_ratio_acv}, and so we prove Theorem~\ref{thm:var_cv_sample_weight}b first.

The proof strategy is to simplify the expression inside the expectation operator in~\eqref{e:var_q_acv}
\begin{equation*}
    \left(\bar{\bm{\mathcal{Q}}} - \bar{\bm{\mu}}^e\right)^{\text{T}} \left[\left(\bm{D}\bm{D}^{\text{T}}\right) \circ \bm{F}^e \right]^{-1}
    \left[ \left(\bm{D}\bm{D}^{\text{T}}\right) \circ \bm{\mathcal{F}}^e_M \circ \left(\bm{\mathcal{F}}^e_M\right)^{\text{T}} \right]
    \left[\left(\bm{D}\bm{D}^{\text{T}}\right) \circ \bm{F}^e \right]^{-1}
    \left(\bar{\bm{\mathcal{Q}}} - \bar{\bm{\mu}}^e\right)
\end{equation*}
using the extra assumptions~\eqref{e:assm_acv_is} and~\eqref{e:assm_acv_mf} on the ACV-IS and ACV-MF schemes. First, an identity is added into the middle term
\begin{equation}
    \begin{aligned}
        \left(\bm{D}\bm{D}^{\text{T}}\right) \circ \bm{\mathcal{F}}^e_M \circ \left(\bm{\mathcal{F}}^e_M\right)^{\text{T}} = \left(\bm{D}\bm{D}^{\text{T}}\right) \circ \bm{F}^e \circ \left(\bm{F}^e\right)^{\circ (-1)} \circ \bm{\mathcal{F}}^e_M \circ \left(\bm{\mathcal{F}}^e_M\right)^{\text{T}}
        \label{e:middle_term}
    \end{aligned}
\end{equation}
where $\left(\bm{F}^e\right)^{\circ (-1)}$
is the Hadamard inverse of $\bm{F}^e$, i.e.,
\begin{equation}
    \begin{aligned}
        \left[\left(\bm{F}^e\right)^{\circ (-1)}\right]_{ij} = \frac{1}{f^e_{ij}},
    \end{aligned}
\end{equation}
and $[\star]_{ij}$ denote an entry of the matrix $\star$. We recall from~\eqref{e:elem_acv_is} and~\eqref{e:elem_acv_mf} that $f^e_{ij}$ depends on the ratios $r_i$ and $r_j$ which must be positive for any meaningful settings. Eventually, using either ACV-IS or ACV-MF makes $r_i$ and $r_j$ greater than 1. Thus, in practice $f^e_{ij} \neq 0$ and the Hadamard inverse of $\bm{F}^e$ exists. 

For the ACV-IS scheme we then have the terms
\begin{equation}
    \begin{aligned}
        \left[\left(\bm{F}^{\text{ACV-IS}}\right)^{\circ (-1)} \circ \bm{\mathcal{F}}^{\text{ACV-IS}}_M \circ \left(\bm{\mathcal{F}}^{\text{ACV-IS}}_M\right)^{\text{T}}\right]_{ij} &= \frac{1}{f^{\text{ACV-IS}}_{ij}} f^{\text{ACV-IS}}_{ii} f^{\text{ACV-IS}}_{jj} = \frac{1}{f^{\text{ACV-IS}}_{ii} f^{\text{ACV-IS}}_{jj}} f^{\text{ACV-IS}}_{ii}f^{\text{ACV-IS}}_{jj} = 1 \\
        \left[\left(\bm{F}^{\text{ACV-IS}}\right)^{\circ (-1)} \circ \bm{\mathcal{F}}^{\text{ACV-IS}}_M \circ \left(\bm{\mathcal{F}}^{\text{ACV-IS}}_M\right)^{\text{T}}\right]_{ii} &= \frac{1}{f^{\text{ACV-IS}}_{ii}} f^{\text{ACV-IS}}_{ii} f^{\text{ACV-IS}}_{ii} = f^{\text{ACV-IS}}_{ii} = \frac{r_i-1}{r_i}
    \end{aligned}
\end{equation}
Because we assume $r_i \gg 1$, 
\begin{equation}
    \begin{aligned}
        \left[\left(\bm{F}^{\text{ACV-IS}}\right)^{\circ (-1)} \circ \bm{\mathcal{F}}^{\text{ACV-IS}}_M \circ \left(\bm{\mathcal{F}}^{\text{ACV-IS}}_M\right)^{\text{T}}\right]_{ii} &= 1, \\
        \left(\bm{F}^{\text{ACV-IS}}\right)^{\circ (-1)} \circ \bm{\mathcal{F}}^{\text{ACV-IS}}_M \circ \left(\bm{\mathcal{F}}^{\text{ACV-IS}}_M\right)^{\text{T}} &= \bm{1}_M \otimes \bm{1}_M.
        \label{e:f_fm_fm_is}
    \end{aligned}
\end{equation}
For the ACV-MF scheme,
\begin{equation}
    \begin{aligned}
        \left[\left(\bm{F}^{\text{ACV-MF}}\right)^{\circ (-1)} \circ \bm{\mathcal{F}}^{\text{ACV-MF}}_M \circ \left(\bm{\mathcal{F}}^{\text{ACV-MF}}_M\right)^{\text{T}}\right]_{ij} &= \frac{1}{f^{\text{ACV-MF}}_{ij}} f^{\text{ACV-MF}}_{ii} f^{\text{ACV-MF}}_{jj} \\
        &= \dfrac{1}{\dfrac{\min(r_i,r_j)-1}{\min(r_i,r_j)}} \frac{r_i-1}{r_i} \frac{r_j-1}{r_j} 
        = \frac{\max(r_i,r_j)-1}{\max(r_i,r_j)} \\
        \left[\left(\bm{F}^{\text{ACV-MF}}\right)^{\circ (-1)} \circ \bm{\mathcal{F}}^{\text{ACV-MF}}_M \circ \left(\bm{\mathcal{F}}^{\text{ACV-MF}}_M\right)^{\text{T}}\right]_{ii} &= \frac{1}{f^{\text{ACV-MF}}_{ii}} f^{\text{ACV-MF}}_{ii} f^{\text{ACV-MF}}_{ii} = f^{\text{ACV-MF}}_{ii} = \frac{r_i-1}{r_i}
    \end{aligned}
\end{equation}
Because we assume $r_i = r$,
\begin{equation}
    \begin{aligned}
        \left[\left(\bm{F}^{\text{ACV-MF}}\right)^{\circ (-1)} \circ \bm{\mathcal{F}}^{\text{ACV-MF}}_M \circ \left(\bm{\mathcal{F}}^{\text{ACV-MF}}_M\right)^{\text{T}}\right]_{ij(ii)} &= \frac{r-1}{r}, \\
        \left(\bm{F}^{\text{ACV-MF}}\right)^{\circ (-1)} \circ \bm{\mathcal{F}}^{\text{ACV-MF}}_M \circ \left(\bm{\mathcal{F}}^{\text{ACV-MF}}_M\right)^{\text{T}} &= \frac{r-1}{r} \left( \bm{1}_M \otimes \bm{1}_M \right).
        \label{e:f_fm_fm_mf}
    \end{aligned}
\end{equation}
Substitute~\eqref{e:f_fm_fm_is} and~\eqref{e:f_fm_fm_mf} into~\eqref{e:middle_term}
\begin{align}
    \left(\bm{D}\bm{D}^{\text{T}}\right) \circ \bm{\mathcal{F}}^{\text{ACV-IS}}_M \circ \left(\bm{\mathcal{F}}^{\text{ACV-IS}}_M\right)^{\text{T}} &= \left(\bm{D}\bm{D}^{\text{T}}\right) \circ \bm{F}^{\text{ACV-IS}} \label{e:sim_middle_term_IS} \\
    \left(\bm{D}\bm{D}^{\text{T}}\right) \circ \bm{\mathcal{F}}^{\text{ACV-MF}}_M \circ \left(\bm{\mathcal{F}}^{\text{ACV-MF}}_M\right)^{\text{T}} &= \frac{r-1}{r} \left(\bm{D}\bm{D}^{\text{T}}\right) \circ \bm{F}^{\text{ACV-MF}} \label{e:sim_middle_term_MF}
\end{align}
Substitute~\eqref{e:sim_middle_term_IS} and~\eqref{e:sim_middle_term_MF} into~\eqref{e:var_q_acv}
\begin{equation}
    \begin{aligned}
        \vvar{}{\bar{\mathcal{Q}}^e(\ubar{\bm{\alpha}}_e)} = \vvar{}{\mathcal{Q}_n(Y_0)}(1-R^2_e) \left( \frac{1}{K} + a(e)\ee{\tilde{\bm{\mathcal{Q}}}^e}{\left(\bar{\bm{\mathcal{Q}}} - \bar{\bm{\mu}}^e\right)^{\text{T}}
        \left[\left(\bm{D}\bm{D}^{\text{T}}\right) \circ \bm{F}^e \right]^{-1}
        \left(\bar{\bm{\mathcal{Q}}} - \bar{\bm{\mu}}^e\right)} \right)
        \label{e:acv_var_tilde}
    \end{aligned}
\end{equation}
Here the expectation of $\left(\bar{\bm{\mathcal{Q}}} - \bar{\bm{\mu}}^e\right)^{\text{T}}
\left[\left(\bm{D}\bm{D}^{\text{T}}\right) \circ \bm{F}^e \right]^{-1}
\left(\bar{\bm{\mathcal{Q}}} - \bar{\bm{\mu}}^e\right)$ can be computed explicitly because $\bar{\bm{\mathcal{Q}}} - \bar{\bm{\mu}}^e$ has a multivariate normal distribution from~\eqref{e:q_mu_acv_dist} and
\begin{equation}
    \begin{aligned}
    \left(\bar{\bm{\mathcal{Q}}} - \bar{\bm{\mu}}^e\right)^{\text{T}}
    \left[\left(\bm{D}\bm{D}^{\text{T}}\right) \circ \bm{F}^e \right]^{-1}
    \left(\bar{\bm{\mathcal{Q}}} - \bar{\bm{\mu}}^e\right) &= \frac{1}{K(K-1)} \dfrac{\left(\bar{\bm{\mathcal{Q}}} - \bar{\bm{\mu}}^e\right)^{\text{T}}}{\sqrt{\dfrac{1}{K}}} \left[\frac{\left(\bm{D}\bm{D}^{\text{T}}\right) \circ \bm{F}^e}{K-1} \right]^{-1}
    \dfrac{\left(\bar{\bm{\mathcal{Q}}} - \bar{\bm{\mu}}^e\right)}{\sqrt{\dfrac{1}{K}}} \\
    &= \frac{1}{K(K-1)} t^2_{M,K-1},
    \label{e:ACV_t2}
    \end{aligned}
\end{equation}
where $t^2_{M,K-1}$ follows the Hotelling's $T^2$ distribution~\cite[Corollary 5.3]{h2019applied}.

Substituting~\eqref{e:ACV_t2} into~\eqref{e:acv_var_tilde}, the variance of $\bar{\mathcal{Q}}^e(\ubar{\bm{\alpha}}_e)$ becomes
\begin{equation}
    \begin{aligned}
    \vvar{}{\bar{\mathcal{Q}}^e(\ubar{\bm{\alpha}}_e)} &= \vvar{}{\mathcal{Q}_n(Y_0)}(1-R^2_e) \left( \frac{1}{K} + \frac{a(e)}{K(K-1)}\ee{\tilde{\bm{\mathcal{Q}}}^e}{t^2_{M,K-1}} \right) \\
    &= \vvar{}{\mathcal{Q}_n(Y_0)}(1-R^2_e) \left( \frac{1}{K} + \frac{a(e)}{K(K-1)}\frac{(K-1)M}{K-M-2} \right) \\
    &= \frac{\vvar{}{\mathcal{Q}_n(Y_0)}}{K}(1-R^2_e) \left( 1 + \frac{a(e)M}{K-M-2} \right),
    \end{aligned}
\end{equation}
where the second equality uses the expectation of the Hotelling's $T^2$ distribution and the third equality simplifies the result.

Thus,
\begin{align*}
    \vvar{}{\bar{\mathcal{Q}}^e(\ubar{\bm{\alpha}}_e)} &= \frac{\vvar{}{\mathcal{Q}^e(\ubar{\bm{\alpha}}_e)}}{K} = \frac{\vvar{}{\mathcal{Q}_n(Y_0)}}{K}(1-R^2_e) \left( 1 + \frac{a(e)M}{K-M-2} \right) \\
    \frac{\vvar{}{\mathcal{Q}^e(\ubar{\bm{\alpha}}_e)}}{\vvar{}{\mathcal{Q}_n(Y_0)}} &= (1-R^2_e) \left( 1 + \frac{a(e)M}{K-M-2} \right)
\end{align*}
To prove Theorem~\ref{thm:var_cv_sample_weight}a, we replace $\bar{\mathcal{Q}}^e(\ubar{\bm{\alpha}}_e)$, $R^2_e$, $\tilde{\bm{\mathcal{Q}}}^e$ and $\bar{\bm{\mu}}^e$ with $\bar{\mathcal{Q}}^{\text{CV}}(\ubar{\bm{\alpha}}_{\text{CV}})$, $R^2$, $\bm{\mathcal{Q}}$ and $\bar{\bm{\mu}}^{\text{CV}}$, respectively, and note that $\bm{F}^e = \bm{\mathcal{F}}^e_M = \bm{1}_M \otimes \bm{1}_M$ in the CV case; hence,
\begin{align}
    \frac{\vvar{}{\mathcal{Q}^{\text{CV}}(\ubar{\bm{\alpha}}_{\text{CV}})}}{\vvar{}{\mathcal{Q}_n(Y_0)}} = (1-R^2) \left( 1 + \frac{M}{K-M-2} \right).
\end{align}
\section{Proof of Corollary~\ref{thm:lower_bounds}}
\label{app:num_samples}
We seek to find $K$ such that Theorem~\ref{thm:var_cv_sample_weight} guarantees variance reduction
\begin{align}
    &\frac{\vvar{}{\mathcal{Q}^e(\ubar{\bm{\alpha}}_e)}}{\vvar{}{\mathcal{Q}_n(Y_0)}} < y \;\text{for}\;0 < y \leq 1 \text{ and } e \in \{\text{CV},\text{ACV-IS},\text{ACV-MF}\}.
    \label{e:num_samples}
\end{align}
We first solve~\eqref{e:num_samples} for $K$ that satisfies this inequality for the ACV-IS and ACV-MF strategies. We then deduce the corresponding result in the CV case. Thus, for $0 < y \leq 1$ and $e \in \{\text{ACV-IS},\text{ACV-MF}\}$, \eqref{e:num_samples} becomes
\begin{align}
    &\left( 1 + \frac{a(e)M}{K-M-2} \right) (1-R^2_e) < y \iff \frac{(K-M-2) + a(e)M}{K-M-2}(1-R^2_e) < y
    \label{e:ineq_K_1}
\end{align}
Because of the assumption $K > M + 2$, \eqref{e:ineq_K_1} becomes
\begin{align}
    &((K-M-2) + a(e)M)(1-R^2_e) < y(K-M-2) \nonumber \\
    &\iff K(1-y-R^2_e) < (M+2-a(e)M)(1-R^2_e)-y(M+2) \nonumber \\
    &\iff K(1-y-R^2_e) < (M+2)(1-y-R^2_e)-a(e)M(1-R^2_e) \nonumber \\
    &\iff \begin{cases}
        K > M+2 - \dfrac{a(e)M(1-R^2_e)}{1-y-R^2_e} &\mbox{if } y+R^2_e > 1 \\
        K < M+2 - \dfrac{a(e)M(1-R^2_e)}{1-y-R^2_e} &\mbox{if } y+R^2_e < 1
    \end{cases}
    \label{e:ineq_K_2}
\end{align}
Now we show that the case $y + R^2_e < 1$ leads to a contradiction. Specifically, ~\eqref{e:ineq_K_2} implies that $M+2 < K < M+2 - \dfrac{a(e)M(1-R^2)}{1-y-R^2_e}$ so that 
\begin{equation}
   \begin{aligned}
        M+2 < M+2 - \dfrac{a(e)M(1-R^2_e)}{1-y-R^2_e} \iff 0 > \dfrac{a(e)M(1-R^2_e)}{1-y-R^2_e} \iff a(e)(1-R^2_e) < 0
    \end{aligned}
    \label{e:ineq_K_3}
\end{equation}
Since in the ACV-MF scheme the low-fidelity models always use more samples than the high-fidelity model~\cite{gorodetsky_generalized_2020}, then $a(e) > 0$ for $e \in \{\text{ACV-IS},\text{ACV-MF}\}$. Therefore, \eqref{e:ineq_K_3} becomes $1-R^2_e < 0$, which is a contradiction to $0 \leq R^2_e \leq 1$.

We are left with $y + R_{e}^2 > 1$. So, we only have $K > B_e = M+2 - \dfrac{a(e)M(1-R^2_e)}{1-y-R^2_e}$ if $y+R^2_e > 1$ and $K > M+2$. In other words, if $y+R^2_e > 1$ and $K > \max (M+2, B_e )$, then the target variance reduction is obtained.

Substituting the values of $a(e)$ into $B_e$, we obtain
\begin{align*}
    B_{\text{ACV-IS}} &= M+2 - \dfrac{M(1-R^2_{\text{ACV-IS}})}{1-y-R^2_{\text{ACV-IS}}}, \\
    B_{\text{ACV-MF}} &= M+2 - \dfrac{(r-1)M(1-R^2_{\text{ACV-MF}})}{r(1-y-R^2_{\text{ACV-MF}})}.
\end{align*}
In particular, if $y=1$, then
\begin{align*}
    B_{\text{ACV-IS}} &= M+2 + \dfrac{M(1-R^2_{\text{ACV-IS}})}{R^2_{\text{ACV-IS}}} = \frac{M}{R^2_{\text{ACV-IS}}} + 2, \\
    B_{\text{ACV-MF}} &= M+2 + \dfrac{(r-1)M(1-R^2_{\text{ACV-MF}})}{rR^2_{\text{ACV-MF}}} = \frac{r-1}{r}\frac{M}{R^2_{\text{ACV-MF}}} + \frac{M}{r} + 2.
\end{align*}
Similarly, for the CV case, if $y+R^2 > 1$, $K > \max ( M+2, B_{\text{CV}} )$, where
\begin{align*}
    B_{\text{CV}} &= M+2 - \dfrac{M(1-R^2)}{1-y-R^2},
\end{align*}
and if $y=1$,
\begin{align*}
    B_{\text{CV}} &= \frac{M}{R^2} + 2.
\end{align*}

\bibliographystyle{unsrt}
\bibliography{refs}

\end{document}